\documentclass[10pt,conference]{IEEEtran}
\IEEEoverridecommandlockouts
\usepackage{cite}
\usepackage{booktabs}
\usepackage{amsmath,amssymb,amsfonts}
\usepackage{algorithmic}
\usepackage{enumitem}
\usepackage{graphicx}
\usepackage{textcomp}
\usepackage{svg}
\usepackage{xcolor}
\usepackage[obeyspaces]{url}
\usepackage{hyperref}
\usepackage{cleveref}
\usepackage{graphicx}
\usepackage{subcaption}
\usepackage{booktabs}
\usepackage{multirow}
\usepackage{diagbox}
\usepackage{graphicx}
\usepackage{xcolor,colortbl}

\usepackage[ruled,linesnumbered]{algorithm2e}
\SetKw{KwBy}{by}
\usepackage{mathtools}
\DeclarePairedDelimiter\ceil{\lceil}{\rceil}

\def\BibTeX{{\rm B\kern-.05em{\sc i\kern-.025em b}\kern-.08em
    T\kern-.1667em\lower.7ex\hbox{E}\kern-.125emX}}

\usepackage[utf8]{inputenc}
\usepackage[english]{babel}
\usepackage{amsthm}
\usepackage{mathrsfs}
\newtheorem{theorem}{Theorem}
\newtheorem{lemma}{Lemma}

\newtheorem{corollary}{Corollary}

\newcommand{\etal}{\textit{et~al}.}
\newcommand{\secref}[1]{Section~\ref{sec: #1}}
\newcommand{\figref}[1]{Fig.~\ref{fig: #1}}
\newcommand{\tabref}[1]{Table~\ref{tab: #1}}

\newcommand{\defref}[1]{Definition~\ref{def: #1}}

\newcommand{\alglineref}[1]{Line~\ref{line: #1}}

\newcommand{\ith}{$i^\textit{th} \ $}
\usepackage{lastpage}
\usepackage{fancyhdr}

\newcommand{\RomanNumeralCaps}[1]
{\MakeUppercase{\romannumeral #1}}
\usepackage{amsmath}

\newtheorem{definition}{Definition}[section]
\theoremstyle{definition}
\newtheorem{Example}{Example}[]

\newtheorem{observation}{Observation}

\begin{document}
\newcommand{\glp}{\boldsymbol{GLP}}
\newcommand{\elp}{\boldsymbol{ELP}}
\newcommand{\elpmore}[2]{$\boldsymbol{ELP}_{E_#1}(#2)$}

\newcommand{\glpeval}{$$\hat{\mathcal{F}}_{\boldsymbol{GLP}}$$}
\newcommand{\glpevalOneindex}[1]{$\hat{\mathcal{F}}_{\boldsymbol{GLP}_{#1}}$}

\newcommand{\Pite}{$\textbf{P}_{ite}$}
\def \ie{i.e.,~}
\def \eg{e.g.,~}

\title{Optimizing Logical Execution Time Model for Both Determinism and Low Latency}

\author{
\IEEEauthorblockN
{
Sen Wang$^1$,
Dong Li$^1$,
Ashrarul H. Sifat$^{1}$,
Shao-Yu Huang$^2$,
Xuanliang Deng$^{1}$, \\
Changhee Jung$^2$, Ryan Williams$^1$, Haibo Zeng$^1$\\
}
\IEEEauthorblockA{{$^1$Virginia Tech, $^2$Purdue University}}
Email: \{swang666, dongli, ashrar7, xuanliang\}@vt.edu,
\{huan1464, chjung\}@purdue.edu, \{rywilli1, hbzeng\}@vt.edu, 
}
\maketitle

\thispagestyle{plain}
\pagestyle{plain}

\begin{abstract}
The Logical Execution Time (LET) programming model has recently received considerable attention, particularly because of its timing and dataflow determinism. In LET, task computation appears always to take the same amount of time (called the task's LET interval), and the task reads (resp. writes) at the beginning (resp. end) of the interval. 
Compared to other communication mechanisms, such as implicit communication and Dynamic Buffer Protocol (DBP), LET performs worse on many metrics, such as end-to-end latency (including reaction time and data age) and time disparity jitter. Compared with the default LET setting, the flexible LET (fLET) model shrinks the LET interval while still guaranteeing schedulability by introducing the virtual offset to defer the read operation and using the virtual deadline to move up the write operation. Therefore, fLET has the potential to significantly improve the end-to-end timing performance while keeping the benefits of deterministic behavior on timing and dataflow.  

To fully realize the potential of fLET, we consider the problem of optimizing the assignments of its virtual offsets and deadlines. We propose new abstractions to describe the task communication pattern and new optimization algorithms to explore the solution space efficiently. 
The algorithms leverage the linearizability of communication patterns and utilize symbolic operations to achieve efficient optimization while providing a theoretical guarantee. 
The framework supports optimizing multiple performance metrics, and guarantees bounded suboptimality when optimizing end-to-end latency.
Experimental results show that our optimization algorithms improve upon the default LET and its existing extensions and significantly outperform implicit communication and DBP in terms of various metrics, such as end-to-end latency, time disparity, and its jitter.
\end{abstract}

\section{Introduction}

Logical Execution Time (LET) is a programming model suitable for developing real-time control applications~\cite{Henzinger2001GiottoAT}. The basic idea behind LET is to abstract away the scheduling and implementation details of software tasks to facilitate model-level, platform-independent testing and verification. 
This is done by defining a time interval for each task, called LET interval, such that the task computation always appears to take the same amount of time as the LET interval, regardless of how long it actually takes to finish. 
In addition, the task communication happens at the boundary of the LET interval: the task always reads at the beginning of the interval and writes at its end.
In most literatures~\cite{Henzinger2001GiottoAT, Tang2023ComparingCP, Hamann2017CommunicationCD, Pazzaglia2023OptimizingIC, Pazzaglia2021OptimalMA, Pazzaglia2019OptimizingTF, Biondi2018AchievingPM, Kordon2020EvaluationOT}, the LET interval spans from the task's release to its deadline (default LET), though shorter LET intervals (flexible LET) are also possible.

The excellent properties of LET have drawn strong interest from academia as well as the automotive industry~\cite{Pazzaglia2021OptimalMA, Pazzaglia2023OptimizingIC, Biondi2018AchievingPM, Pazzaglia2019OptimizingTF, Shrivastava2021IntroductionTT}. In particular, the time determinism and dataflow determinism of LET have made it an appealing solution as a communication mechanism, especially for multicore platforms~\cite{ZiegenbeinH15}, as these properties simplify the analysis of systems' temporal behavior, stabilize end-to-end latency variance~\cite{Martinez2018AnalyticalCO}, and make the system's temporal behavior composable~\cite{Hamann2017CommunicationCD}. Here, time determinism refers to the fact that communication happens only at predefined time instants. Dataflow determinism refers to the property that in a communication link, a consumer job always reads the output value from the same producer job. 

However, as a communication mechanism, LET also introduces significant drawbacks regarding end-to-end latency metrics, including reaction time and data age. These metrics are important for real-time systems' safety. For example, it is usually desirable that the real-time computing system reacts to sensor readings or external events as quickly as possible, imposing a system design that minimizes the reaction time~\cite{Percept21_RTSS}. The default LET only makes the output data available at the job deadline, even if the data may be ready much earlier. 
Hence, compared with other popular communication mechanisms such as implicit communication and Dynamic Buffer Protocol, the default LET is proven to have a longer data age or reaction time~\cite{Tang2023ComparingCP}. 

The flexible LET (fLET) model~\cite{kirsch2006evolution, Biondi2018AchievingPM, Bini2023ZeroJitterCO, autosar_timing_extensions22} provides a potential solution by adjusting the LET interval. 
The difference between fLET and the default LET is that the release time of a task is delayed by a \textit{virtual offset}, and the reading happens deterministically at this deferred release time (Some literatures~\cite{kirsch2006evolution} only delay the reading time). Similarly, a task's deadline and the write operation are brought forward to an earlier \textit{virtual deadline}. We provide more examples of fLET later in the paper (\eg Fig.~\ref{Fig_LET_fLET_IC_Example}).
The flexibility of fLET provides big potential for performance improvements (as shown in our experiments).
However, how to optimize both the virtual offset and virtual deadline while offering a theoretical guarantee is missing in the literature.

\noindent$\textbf{Contributions.}$ 
\begin{itemize}[leftmargin=*]
   \item Propose novel optimization algorithms to optimize \textit{both} the LET intervals' start and finish times.
   \item Perform symbolic operations of a newly introduced concept, \textit{communication pattern}, which compactly models tasks' possible reading/writing relationships.
   \item Establish suboptimality bounds when minimizing the data age or reaction time. To the best of our knowledge, this is the first work that provides the theoretical guarantee for fLET while maintaining fast runtime speed.
   \item Support minimizing various latency metrics for fLET. To the best of our knowledge, this is the first optimization algorithm that supports optimizing time disparity and jitter.
   \item To the best of our knowledge, this is the first work that shows LET after optimization outperforms alternative communication protocols (implicit communication and DBP) across multiple metrics, such as end-to-end latency, time disparity, and jitter. Our findings enhance fLET's competitiveness in broader scenarios.
\end{itemize}

\noindent \textbf{Paper Organization.} 
\secref{relatedwork} summarizes the related work. 
\secref{system_model} introduces the system model and reviews the latency metrics and task communication mechanisms. Section~\ref{section_fLET} describes the flexible LET (fLET) model. Section~\ref{section_commu_pattern} defines the communication pattern, a new abstraction for task communication that will be leveraged in our optimization algorithms. Section~\ref{section_optimization} gives a general optimization framework for optimizing any metric, while Section~\ref{section_dart_optimization} proposes new theorems to optimize data age and reaction time efficiently. Section~\ref{section_application} presents generalizations, limitations, and possible solutions. We present the experimental results in Section~\ref{experiment_section} and conclude the paper in Section~\ref{conclusion_section}.

\section{Related Work}
\label{sec: relatedwork}

The RTSS2021 industry challenge~\cite{Percept21_RTSS} comprehensively summarizes essential latency metrics: data age, reaction time, and time disparity.
Among them, data age and reaction time (DART) have been studied extensively for analysis and optimization~\cite{Feiertag2008ACF, Abdullah2019WorstCaseCR, Gnzel2021TimingAO, Drr2019EndtoEndTA, Schlatow2018DataAgeAA, Kloda2018LatencyAF, Verucchi2020LatencyAwareGO, Klaus2021ConstrainedDW, Gnzel2023OnTE}. 
As for time disparity, Li~\etal~\cite{Li2022WorstCaseTD} developed analysis methods for the Robotic Operating System~\cite{Li2022WorstCaseTD}, Jiang~\etal~\cite{Jiang2023AnalysisAO} proposed analysis and buffer size design to reduce the worst-case time disparity. However, most of these studies focus on communication mechanisms other than LET or its variants. 

The Logical Execution Time (LET) programming model, introduced by Henzinger \etal~\cite{Henzinger2001GiottoAT} with the programming language Giotto, has gained attention from the automotive industry~\cite{Biondi2018AchievingPM, Ernst2018TheLE, Hamann2017CommunicationCD, Pazzaglia2023OptimizingIC, Martinez2018AnalyticalCO, Gemlau21TCPS}. It is also integrated into the automotive software architecture standard AUTOSAR~\cite{Tang2023ComparingCP, autosar_timing_extensions22}. 
The popularity in the automotive domain inspires related research in processor assignments~\cite{Pazzaglia2019OptimizingTF, Pazzaglia2021OptimalMA}, end-to-end latency analysis~\cite{Martinez2018AnalyticalCO, Kordon2020EvaluationOT, Martinez2020EndtoendLC, Becker2017EndtoendTA}, and optimization~\cite{Martinez2018AnalyticalCO, Lee2022GeneralizingLE, Bradatsch2016DataAD}. 
However, most research on LET assumes the default LET model (LET intervals equal periods), which performs worse in end-to-end latency metrics than implicit communication and DBP~\cite{Tang2023ComparingCP} and may suffer from larger latency jitter.

Variants of the flexible LET model and heuristics algorithms have been considered in the literature to reduce the end-to-end latency. Martinez~\etal~\cite{Martinez2018AnalyticalCO} proposed to assign an offset to tasks upon their first release while keeping the LET intervals the same as the default LET. 
Bradatsch \etal~\cite{Bradatsch2016DataAD} proposed using the worst-case response time as the length of the LET interval, but the reading time remains the same as the default LET model (\ie no offset).
Maia~\etal~\cite{Maia2023ReducingEL} set the LET interval of each task from the smallest relative start time to the largest relative finish time based on schedules. 
Apart from the fLET model modifications, heuristic optimization algorithms have also been proposed to reduce end-to-end latency~\cite{Martinez2018AnalyticalCO, Maia2023ReducingEL}. 
However, these algorithms may suffer from scalability issues or cannot provide an optimality guarantee for fLET.

Offset assignments could reduce the jitter of many metrics, such as response time~\cite{Tindell1994ADDINGTT, Palencia1998SchedulabilityAF} and end-to-end latency~\cite{Martinez2018AnalyticalCO}. Bini~\etal~\cite{Bini2023ZeroJitterCO} proposed to utilize ring algebra to eliminate the jitter of end-to-end latency in LET models. However, the jitter of the time disparity metric is rarely studied.

Compared to previous LET-related studies, this paper \emph{addresses a more general optimization problem in terms of variables or objective functions. It also introduces novel optimization techniques with improved performance and stronger theoretical guarantees}.

\section{System Model and Definitions}
\label{sec: system_model}
This paper uses bold fonts to represent a vector or a set and light characters for scalars. The \ith element of a set $\boldsymbol{S}$ is denoted as $\boldsymbol{S}(i)$. 
We may use sub- and super-scripts to distinguish notations, such as $q_j^n$.


\subsection{System Model}
This paper considers a task set $\boldsymbol{\tau}$ of periodic tasks (Generalization to sporadic tasks will be discussed later). Each task $\tau_i$ is described by a tuple $(C_i, T_i, D_i^{org})$, which denotes the worst-case execution time, period, and relative deadline, respectively. 
We assume $D_i^{org} \leq T_i$.
The least common multiple of all the task periods is hyper-period $H$. 
Every $T_i$ time units, $\tau_i$ releases a job for execution. The $q_i$-th released job of the task $\tau_i$ is usually denoted as $J_{i,q_i}$. 
We allow the job index $q_i$ to be negative numbers. For example, $J_{i, -1}$ is released one period earlier than $J_{i, 0}$. To simplify notations, we assume the $0^{th}$ jobs of all the tasks become ready at time 0. However, our optimization algorithms can optimize LET intervals if there is a known release offset.

We assume some worst-case response time analyses (RTA) are available. The system is schedulable if the RTA is no larger than the relative deadline.
Ideally, the worst-case RTA should have no dependency or linear dependency on other tasks' release time and deadline (an example is provided in Section~\ref{experiment_section}). 
Otherwise, please refer to Section~\ref{section_more_application_schedulability}.

The input of a task $\tau_j$ (or a job $J_{j,q_j}$) could depend on the output of a task $\tau_i$ (or a job $J_{i,q_i}$), which we denote as $\tau_i \rightarrow \tau_j$ (or $J_{i,q_i} \rightarrow J_{j,q_j}$). 
The dependency relationships of all the tasks in $\boldsymbol{\tau}$ can be represented as a directed acyclic graph (DAG) $\boldsymbol{G}=(\boldsymbol{\tau}, \boldsymbol{E})$, where $\boldsymbol{E}$ denotes all the edges in the graph. 
The DAG is not assumed to be fully connected for generality.
A cause-effect chain $\mathcal{C}$ is a list of tasks $\{\tau_0$, $\tau_1$ ... $\tau_n \}$ where $\tau_i \rightarrow \tau_{i+1}$. 
These tasks are typically indexed based on their causal relationships: $\tau_0$ is usually the source task that depends on no other tasks in $\mathcal{C}$, $\tau_n$ is the sink task that no other tasks in $\mathcal{C}$ depend on. 
Multiple cause-effect chains may share some tasks in a graph $\boldsymbol{G}$. All the cause-effect chains are denoted as $\boldsymbol{\mathcal{C}}$.
A job $J_{i, q_i}$'s reading (writing) time is denoted as $re_{J_{i,q_i}}$ ($wr_{J_{i,q_i}}$).

\subsection{Data Age and Reaction Time}
Data age and reaction time (DART) are among the most commonly used metrics to measure end-to-end latency.
Given a cause-effect chain $\mathcal{C}=\tau_0 \rightarrow... \rightarrow \tau_n$, data age measures the longest time a sensor event influences a computation system. In contrast, reaction time measures the maximum delay between the occurrence of the first sensor event and the last event of the chain that depends on the sensor event. 
Next, we review the DART analysis following G{\"u}nzel~\etal~\cite{Gnzel2021TimingAO}:
\begin{definition}[Job chain~\cite{Gnzel2021TimingAO}]
    A job chain $\mathcal{C}^J$ of a cause effect chain $\mathcal{C}$ is a sequence of jobs $(J_{0,q_0},...,J_{n,q_n})$ where the reading data of job $J_{i+1, q_{i+1}}$ is generated by $J_{i, q_{i}}$.
\end{definition}
\begin{definition}[Length of a job chain, $\text{Len}(\mathcal{C}^J)$]
\label{def_length_job_chain}
    The length of a job chain $\mathcal{C}^J=(J_{0,q_0},..., J_{n,q_n})$ is $wr_{J_{n,q_n}}-re_{J_{0,q_0}}$.
\end{definition}

When the reader task and the writer task have different periods, issues of under-sampling may arise, where a writer's data cannot reach any reader due to being overwritten by subsequent writers. Conversely, over-sampling may occur, wherein a writer's data is read by multiple readers simultaneously~\cite{Abdullah2019WorstCaseCR}. The following concepts aid in providing a more precise description of these communication relationships.

\begin{definition}[Last-reading job]
   \label{def: last_reading_job}
    Consider two tasks $\tau_i \rightarrow \tau_j$. For any jobs $J_{j, q_j}$, it has a unique last-reading job $J_{i, \overleftarrow{q_j}}$ that satisfies the following properties:
    \begin{equation}         
    wr_{J_{i,\overleftarrow{q_j}}} \leq re_{J_{j,q_j}} < wr_{J_{i,\overleftarrow{q_j}+1}}
        \label{last_reading_job}
    \end{equation}
\end{definition}
For example, in Fig.~\ref{fig_stanLET_Example}, $J_{1,0}$ is $J_{2,3}$'s last-reading job.

\begin{definition}[Immediate backward job chain, \cite{Gnzel2021TimingAO}]
    An immediate backward job chain is a job chain $(J_{0,q_0},...,J_{n,q_n})$ where each $J_{i, q_{i}}$ is the last-reading job of $J_{i+1, q_{i+1}}$.
\end{definition}
\begin{definition}[First-reacting job]
   \label{def: first_reacting_job}
    Consider two tasks $\tau_i \rightarrow \tau_j$. For any jobs $J_{i, q_i}$, it has a unique first-reacting job $J_{j, \overrightarrow{q_i}}$ which satisfies the following properties:
    \begin{equation}
    re_{J_{j,\overrightarrow{q_i}-1}} < wr_{J_{i,q_i}} \leq re_{J_{j,\overrightarrow{q_i}}}
        \label{first_reacting_job}
    \end{equation}
\end{definition}
For example, in Fig.~\ref{fig_stanLET_Example}, $J_{1,1}$ is $J_{0,0}$'s first-reacting job.

\begin{definition}[Immediate forward job chain, \cite{Gnzel2021TimingAO}]
    An immediate forward job chain is a job chain $(J_{0,q_0},...,J_{n,q_n})$ where each $J_{i+1, q_{i+1}}$ is the first-reacting job of the job $J_{i, q_{i}}$.
\end{definition}
In this paper, We use upper arrows to distinguish one job's last-reading or first-reacting job as in \defref{last_reading_job} and \defref{first_reacting_job}.
Following~\cite{Gnzel2021TimingAO}, the worst-case data age (reaction time) of a cause-effect chain is the length of its longest immediate backward (forward) job chains.

\begin{figure}[t!]
\centering
\begin{subfigure}[b]{0.9\columnwidth} 
    \includegraphics[width=\textwidth]{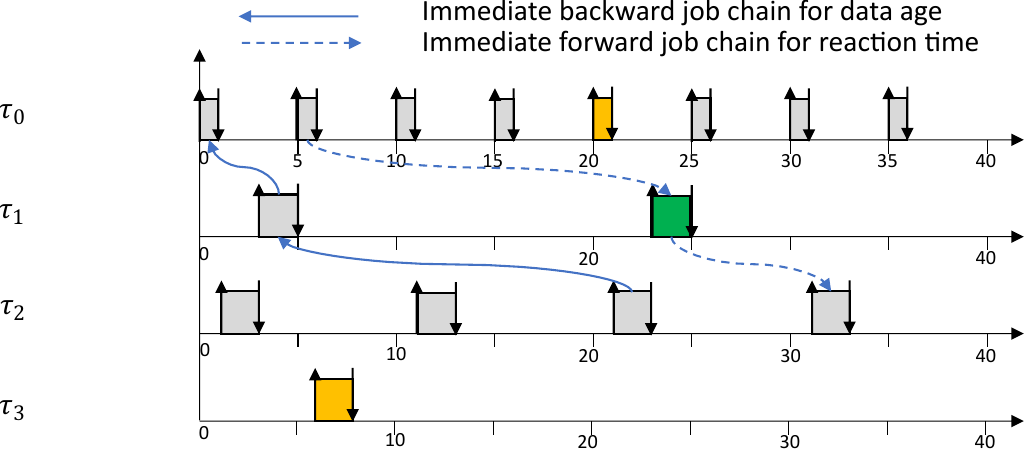}
    \caption{Implicit communication: RT=28, DA=23, TD=33, Jitter=20}
    \label{fig: implicit_example}
\end{subfigure}
\begin{subfigure}[b]{0.9\columnwidth}
    \includegraphics[width=\textwidth]{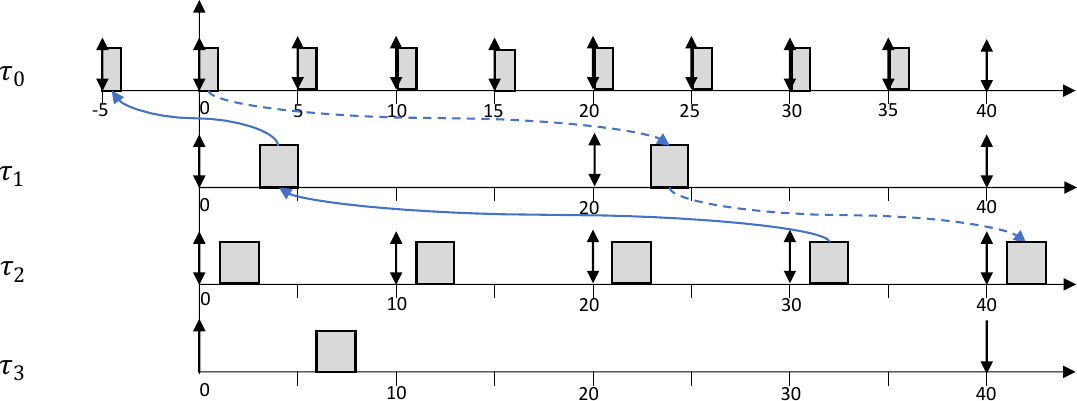}
    \caption{Default LET: RT=50, DA=45, TD=20, Jitter=20  }
    \label{fig_stanLET_Example}
\end{subfigure}
\begin{subfigure}[b]{0.9\columnwidth}
\includegraphics[width=\textwidth]{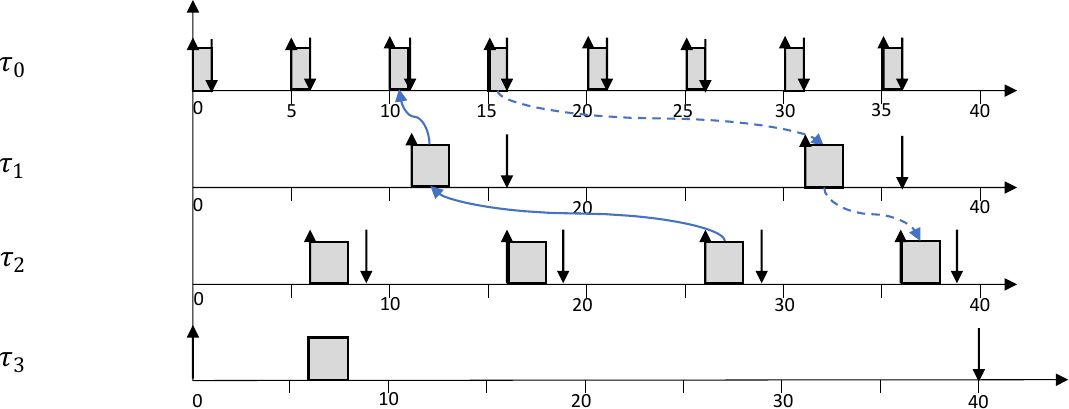}
    \caption{fLET-Optimize-DA, RT=\textbf{24}, DA=\textbf{19}, TD=31, Jitter=20}
    \label{fig_fLET_Example}
\end{subfigure}
\begin{subfigure}[b]{0.9\columnwidth}
\includegraphics[width=\textwidth]{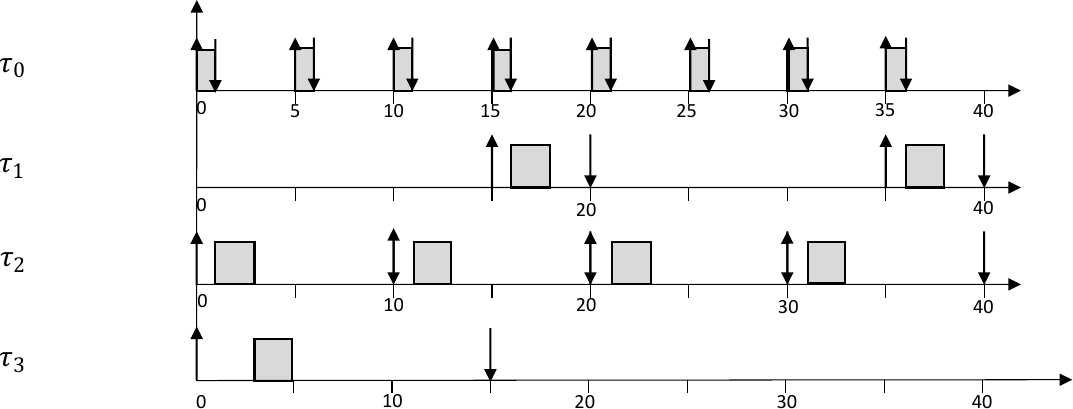}
    \caption{fLET-Optimize-TD-Jit, RT=35, DA=30, TD=\textbf{16}, Jitter=\textbf{12} }
    \label{fig_fLET_Example_sf_jitter}
\end{subfigure}
\caption{
The schedules of different communication protocols in Example~\ref{Example_setup}.
The upward and downward arrows represent job reading and writing times, respectively. 
Solid leftward arrows connect the longest immediate backward job chains, and dashed rightward arrows connect the longest immediate forward job chains.
It is worth highlighting that fLET demonstrates superior performance compared to alternative methods after optimization. Figures \ref{fig_fLET_Example} and \ref{fig_fLET_Example_sf_jitter} serve merely as illustrations of the optimization results for different metrics. Importantly, our optimization approach supports the simultaneous optimization of multiple metrics.
}
\label{Fig_LET_fLET_IC_Example}
\end{figure}

\begin{Example}
\label{Example_setup}
Consider a DAG with four tasks: $\{\tau_0, \tau_1, \tau_2, \tau_3\}$ as shown in \figref{example_dag_fig}. The four tasks are scheduled on one CPU with Rate Monotonic preemptively.
The objective is minimizing the worst-case data age of the cause-effect chain $\mathcal{C}_0:\{\tau_0 \rightarrow \tau_1 \rightarrow \tau_2\}$ and $\mathcal{C}_1:\{\tau_3 \rightarrow \tau_1 \rightarrow \tau_2\}$.

\begin{figure}[t!]
\centering
\includegraphics[width=1.0\columnwidth]{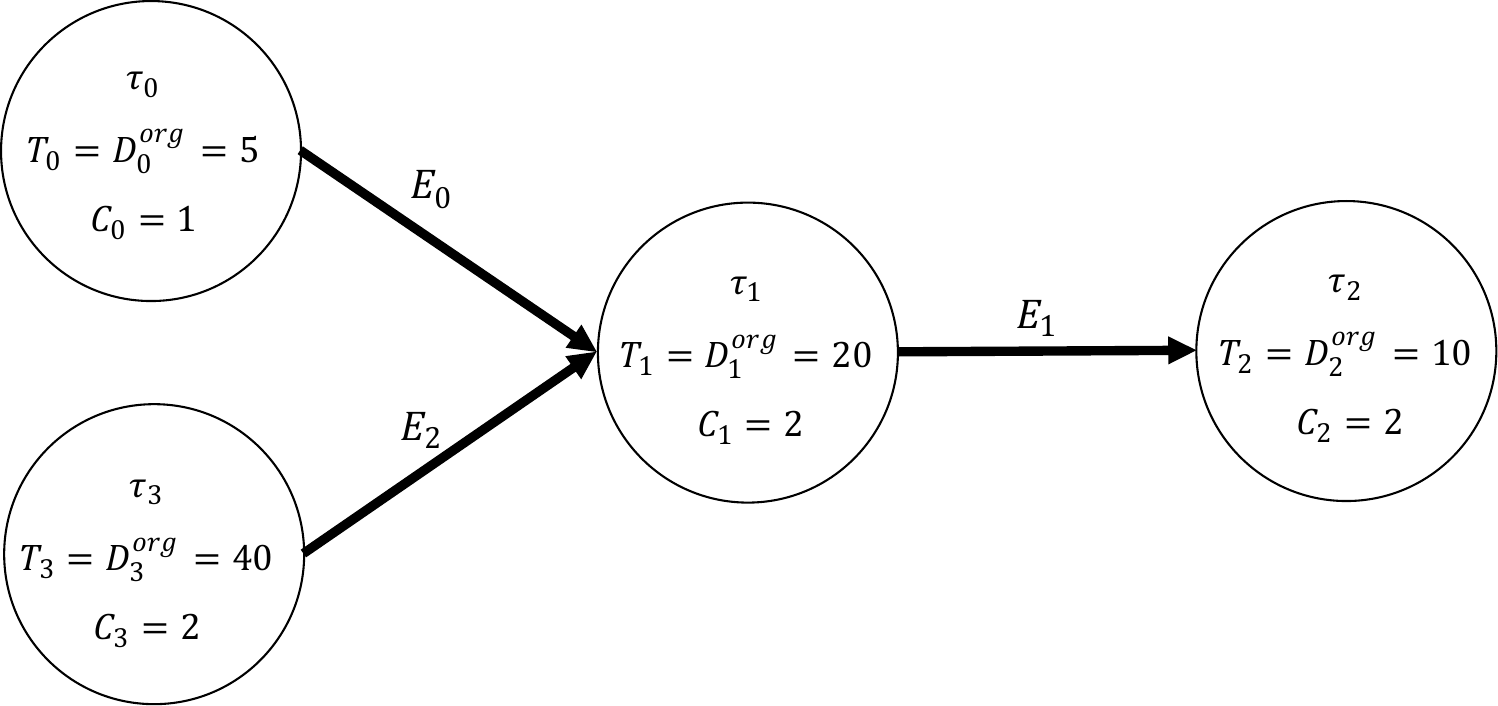}
\caption{Example DAG with two cause-effect chains.}
\label{fig: example_dag_fig}
\end{figure}

Consider the cause-effect chain $\mathcal{C}_0$. 
In Fig.~\ref{fig_stanLET_Example},
$J_{0,0} \rightarrow$ $J_{1,1}$ $\rightarrow J_{2,4}$ formulate an immediate forward job chain (length is 50, notice $wr_{J_{2,4}}=50$), $J_{0,-1}\rightarrow J_{1,0} \rightarrow J_{2,3}$ formulate an immediate backward job chain (length is 45). 
\end{Example}

There are alternative definitions of data age and reaction time~\cite{Gnzel2021TimingAO, Gunzel23TECS}.
Our optimization algorithm is compatible with them as they have a similar analysis.

\subsection{Time disparity}
\label{Section_SF}
In many situations, such as autonomous driving, a task's input could depend on multiple source tasks (the sink and source tasks constitute a \textit{merge}). In this case, the input data must be collected at approximately the same time to be useful.
\begin{definition}[Time disparity, \cite{Percept21_RTSS}]
\label{def_sf}
    Consider a task $\tau_j$ which reads input from several source tasks $\boldsymbol{\tau}^s$ ($\tau_j$ and $\boldsymbol{\tau}^s$ constitute a ``merge''). For each job, $J_{j,q_j}$, denote its last-reading jobs from $\boldsymbol{\tau}^{s}$ as $\boldsymbol{J}^{lr,s}$.
    The time disparity of $J_{j,q_j}$ is defined as the difference between the earliest and the latest writing time of the jobs in $\boldsymbol{J}^{lr,s}$:
    \begin{equation}
      TD(J_{j,q_j}) =  \max_{J \in \boldsymbol{J}^{lr,s} }wr_J - 
        \min_{J \in \boldsymbol{J}^{lr,s}} wr_J
        \label{time_disparity_def}
    \end{equation}
\end{definition}

\begin{Example}
    In Example~\ref{Example_setup}, $\tau_1$, $\tau_0$ and $\tau_3$ formulate a merge. 
    In \figref{implicit_example}, the job {$J_{1,1}$} (highlighted in green) reads data from {$J_{0,4}$} (yellow) and {$J_{3,0}$} (yellow). The time disparity of $J_{1,1}$ is $21-8=13$; Similarly, $J_{1,0}$'s time disparity is $1-(-32)=33$. $\tau_1$'s worst-case time disparity is $33$; its jitter is $33-13=20$.
\end{Example}

\subsection{Task Communication Mechanisms}
In this section, we briefly review several popular communication mechanisms following~\cite{Hamann2017CommunicationCD, Tang2023ComparingCP}:
\begin{itemize}[leftmargin=*]
    \item Direct Communication: allows unrestricted access to shared variables, enabling tasks to read or write data whenever needed. It may suffer from data inconsistency issues~\cite{Hamann2017CommunicationCD}.
    \item Implicit communication:
    Tasks read input data when they begin execution and write output data when they finish.
   This approach depends on task schedules, making it more complex to analyze and optimize system behavior, such as end-to-end latency.
\item Default Logical Execution Time: The timing of data read and write is deterministic if the system is schedulable, independent of task schedules. 
Typically,  the reading and writing times are defined as follows:
\begin{equation}
    \forall J_{i,q_i},\ re_{J_{i,q_i}}= q_iT_i, \ wr_{J_{i,q_i}}= q_iT_i + D_i^{org}
    \label{eq_read_write_standard_let}
\end{equation}
    \item Dynamic Buffer Protocol (DBP): 
    A non-blocking buffering protocol that can preserve synchronous semantics under preemptive scheduling \cite{sofronis2006memory}. Instead of restricting data reads and writes at the beginning or end of task execution, DBP allows arbitrary data access through task execution without compromising data consistency by managing multiple buffer copies for the same data. However, there may be additional memory and buffer management overhead. 
\end{itemize}
Under the implicit communication protocol, the length of a cause-effect chain, measured by the maximum data age or reaction time, is no longer than under the DBP communication protocol, which is, in turn, no longer than under the default LET communication protocol~\cite{Tang2023ComparingCP}. 

\section{ \texorpdfstring{The Flexible LET Model and \\ Problem Description}{Problem Description}}
\label{section_fLET}

This section reviews the flexible LET model~\cite{kirsch2006evolution, Biondi2018AchievingPM, autosar_timing_extensions22} where tasks can read input and write output at times different from those specified in equation~\eqref{eq_read_write_standard_let}:
\begin{definition}[Virtual offset, $O_i$]
\label{def_virtual_offset}
    For each job $J_{i,q_i}$ of a periodic task $\tau_i$, both its release time and the time it reads its input, $re_{J_{i,q_i}}$, is delayed from $q_iT_i$ to $ O_i + q_iT_i$, where $ O_i \geq 0$.
\end{definition}
\begin{definition}[Virtual deadline, $D_i$]
\label{def_virtual_deadline}
     For each job $J_{i,q_i}$ of a periodic task $\tau_i$, both its deadline and the time to write its output, $wr_{J_{i,q_i}}$, is moved forward from $q_iT_i+D_i^{org}$ to $q_iT_i+D_i$, where $O_i \leq D_i \leq D_i^{org}$.
\end{definition}
In the flexible LET model, all jobs of the same task share the same virtual offset and virtual deadline, ensuring deterministic reading and writing times.

\subsection{fLET vs Alternative Communication Mechanism}
Compared to alternative methods, such as default LET and implicit communication (both are used in AUTOSAR~\cite{Tang2023ComparingCP}), fLET shows significant potential in optimizing various performance metrics while still preserving its deterministic advantages. However, realizing these benefits of the flexible LET model depends on finding an optimal set of virtual offsets and virtual deadlines, which is the central focus of this paper.

\subsection{fLET Optimization Problem Descriptions}
This paper aims to find virtual offsets and virtual deadlines to minimize the worst-case data age (or reaction time):
\begin{equation}
\label{obj_overall}
    \min_{\boldsymbol{O},\boldsymbol{D}} \sum_{\mathcal{C} \in \boldsymbol{\mathcal{C}}} \max \boldsymbol{\mathcal{F}}_{\mathcal{C}}(\boldsymbol{O},\boldsymbol{D})
\end{equation}
where $\boldsymbol{O}$ and $\boldsymbol{D}$ denote the virtual offset and virtual deadline for all the tasks in the task set $\boldsymbol{\tau}$; $\boldsymbol{\mathcal{F}}_{\mathcal{C}}$ denotes the length of all the immediate forward or backward job chains. 

Another important metric considered in this paper is the worst-case time disparity and jitter:
\begin{align}
    \min_{\boldsymbol{O},\boldsymbol{D}} \sum_{\mathcal{M} \in \boldsymbol{\mathcal{M}}} 
    & \max \boldsymbol{\mathcal{F}}_{\mathcal{M}}(\boldsymbol{O},\boldsymbol{D}) + \nonumber \\ 
    &  \omega(\max \boldsymbol{\mathcal{F}}_{\mathcal{M}}(\boldsymbol{O},\boldsymbol{D}) - \min \boldsymbol{\mathcal{F}}_{\mathcal{M}}(\boldsymbol{O},\boldsymbol{D}))    
\label{obj_overall_sf}
\end{align}
where $\mathcal{M}$ denotes a ``merge'' that includes a sink task and multiple source tasks that the sink task depends on, $\boldsymbol{\mathcal{F}}_{\mathcal{M}}$ follows Definition~\ref{def_sf}, $\omega$ is the weight of the jitter term. 

Both the optimization problems~\eqref{obj_overall} and~\eqref{obj_overall_sf} are subject to the schedulability constraints:
\begin{equation}
\label{schedulability_test_fLET}
    \forall \tau_i \in \boldsymbol{\tau}, 0 \leq O_i, \ O_i+R_i \leq D_i,\ D_i \leq D_i^{org}
\end{equation}
where $R_i$ denotes the worst-case response time for task $\tau_i$. 

\section{Communication Patterns}
\label{section_commu_pattern}

\begin{table*}[]
\caption{Edge/Graph last-reading Patterns and operations. (First-reacting patterns are defined similarly)}
\label{tab: symbol operations}
\resizebox{\textwidth}{!}{%
\begin{tabular}{@{}llllllll@{}}
\toprule
Symbol &
  Definition &
  \begin{tabular}[c]{@{}l@{}}Mathematical\\ Description\end{tabular} &
  \begin{tabular}[c]{@{}l@{}}Feasibility\\ Check\end{tabular} &
  Evaluation &
  \multicolumn{3}{c}{Other Operations} \\ \midrule
ELP &
  \begin{tabular}[c]{@{}l@{}}Job map,\\ \defref{job map}\end{tabular} &
  \begin{tabular}[c]{@{}l@{}}2 Linear Inequalities\\ \defref{elp inequalities}\end{tabular} &
  \multirow{2}{*}{\begin{tabular}[c]{@{}l@{}}Linear inequalities\\  feasibility check\\ \defref{feasibility}\end{tabular}} &
  N/A &
  \begin{tabular}[c]{@{}l@{}}Add last-reading job pairs,\\ \defref{elp inequalities}\end{tabular} &
  \begin{tabular}[c]{@{}l@{}}ELP Comparison,\\ \defref{elp comp}\end{tabular} &
  N/A \\ \cmidrule(r){1-3} \cmidrule(l){5-8} 
(partial) GLP &
  \begin{tabular}[c]{@{}l@{}}ELP set, \\ \defref{glp}, \defref{partial glp}\end{tabular} &
  \begin{tabular}[c]{@{}l@{}}Multiple Linear \\ Inequalities\end{tabular} &
   &
  \defref{glp evaluate} &
  \begin{tabular}[c]{@{}l@{}}Add ELP,\\ \defref{glp}\end{tabular} &
  \begin{tabular}[c]{@{}l@{}}GLP Comparison,\\ \defref{glp comp}\end{tabular} &
  \begin{tabular}[c]{@{}l@{}}GLP Contain,\\ \defref{contain}\end{tabular} \\ \bottomrule
\end{tabular}%
}
\end{table*}

\subsection{Motivation}
Directly solving the problem~\eqref{obj_overall} and \eqref{obj_overall_sf} is challenging due to the nonlinear and non-continuous nature of these problems, which require iterating through all virtual offset and virtual deadline combinations to find the optimal solution. 
Therefore, this section introduces a new concept, \textit{communication pattern}, which succinctly describes the communication among tasks. 
We then reformulate the problem~\eqref{obj_overall} and~\eqref{obj_overall_sf} to find the optimal communication pattern. 
Operations related to the communication patterns are summarized in \tabref{symbol operations}.

Searching for the communication patterns is easier than value combinations because they have a much smaller solution space. 
Furthermore, efficient feasibility analysis and symbolic operations are available to reduce unnecessary computation substantially.
This method improves both performance and efficiency than the state-of-art~\cite{Martinez2018AnalyticalCO}.

\subsection{Edge Communication Patterns (ECP)}
In the flexible LET model, all the jobs of the same task have the same virtual offset and virtual deadline. 
Therefore, for any edge $E= \tau_i \rightarrow \tau_j$, it is sufficient to focus on jobs within a super-period $H^s$ (the least common multiple of $T_i$ and $T_j$, \cite{Verucchi2020LatencyAwareGO}) to describe the possible reading/writing relationships. 

\begin{definition}[Job map]
\label{def: job map}
    A job map is a mapping from one job $J_{i, q_i}$ in a set to another job $J_{j,q_j}$ in another set.
\end{definition}

\begin{definition}[Edge last-reading pattern, $ELP$]
    Consider an edge $E= \tau_i \rightarrow \tau_j$. For any job $J_{j, q_j}$ where $0 \leq q_j < H^s/T_j \}$, the edge last-reading pattern $ELP$ is a job map which specifies which job of $\tau_i$ is $J_{j, q_j}$'s last-reading job.
\end{definition}
\begin{definition}[Edge first-reacting pattern, $EFP$]
    Consider an edge $E= \tau_i \rightarrow \tau_j$. For any job $J_{i, q_i}$ where $ 0 \leq q_i < H^s/T_i \}$, the edge first-reacting pattern $EFP$ is a job map which specifies which job of $\tau_j$ is $J_{i, q_i}$'s first-reacting job.
\end{definition}

An edge $E$ could have multiple ELPs or EFPs. We use $\boldsymbol{ELP}_{E}(i)$ ($\boldsymbol{EFP}_{E}(i)$) to denote the \ith ELP (EFP) of $E$. Both ELP and EFP are considered ECP.

Consider an edge $E= \tau_i \rightarrow \tau_j$, a job $J_{j, q_j}$ and its last-reading job $J_{i, \overleftarrow{q_j}}$, we can use \defref{last_reading_job} to get:
\begin{equation}
\label{last_reading_job_ineq}
    D_i+ \overleftarrow{q_j} T_i \leq O_j+q_jT_j <D_i+ (\overleftarrow{q_j}+1) T_i 
\end{equation}
Therefore, ELPs can be described by linear inequalities:
\begin{definition}[ELP inequalities]
\label{def: elp inequalities}
    Given an edge $E= \tau_i \rightarrow \tau_j$ and an ELP, the ELP inequalities are the intersection of the linear inequalities (\defref{last_reading_job}) of all the jobs $J_{j, q_j} \in \boldsymbol{J}_j =\{ {J_{j, q_j} | 0 \leq q_j < H^s/T_j} \}$ and their last-reading jobs:
    \begin{equation}
    \mathcal{I}_{ELP} = \mathcal{I}^{LR}_0 \cap  ... \ \mathcal{I}_{q_j}^{LR} ...   \cap  \mathcal{I}^{LR}_{H^s/T_j-1}
    \label{eq_ecp_fr}
\end{equation}
\end{definition}
\noindent where $\mathcal{I}^{LR}_{q_j}$ is the linear inequality~\eqref{last_reading_job_ineq} between $J_{j, q_j}$ and its last-reading job specified by the ELP.
The symbol $ \cap $ denotes that the two linear inequalities must be \textit{both} satisfied. 
Notice that Equation~\eqref{eq_ecp_fr} can be merged into two linear inequalities since all the $\mathcal{I}_{q_j}^{LR}$ depends on the same two variables (\ie $O_j$ and $D_i$).
Similarly, EFP inequalities can be obtained based on first-reacting jobs.

\begin{Example}
    Consider the edge $E_1= \tau_1 \rightarrow \tau_2$ in Example~\ref{Example_setup}. Suppose $J_{2,0}$ and $J_{2,1}$'s last-reading jobs are $J_{1,-1}$ and $J_{1,0}$, respectively, then the linear inequalities associated with the two jobs are
    $D_1-20 \leq O_2 <D_1+0$ and $D_1 \leq O_2 + 10 <D_1+20$.
    Since these inequalities have to hold together, we can merge them into one set of inequalities by taking the intersection $D_1 -10 \leq O_2 < D_1$.
\end{Example}

\noindent \textbf{Usage.} ELPs are used for optimizing data age or time disparity, whereas EFPs are used for optimizing reaction time.

\begin{definition}[Feasibility]
    An ELP (EFP) is feasible if there exists a set of virtual offset and virtual deadline assignments that satisfy its inequality constraints (\defref{elp inequalities}) and schedulability constraints~\eqref{schedulability_test_fLET}.
\end{definition}

\begin{Example}
Consider the edge $E_0:\tau_0 \rightarrow \tau_1$ in example~\ref{Example_setup} and an ELP which specifies $J_{1,0}$'s last-reading job to be $J_{0,3}$. 
It is infeasible because there are no solutions which can satisfy the following inequalities:
\begin{equation*}
    R_0 \leq D_0, D_0 +15 \leq O_1  < D_0+20,
    O_1+R_1 \leq D_1
\end{equation*}
where $R_0=1$, $R_1=5$ based on RM.
\end{Example}

The following lemmas provide the necessary conditions for an edge communication pattern to be feasible.

\begin{lemma}
\label{lemma_two_task_Reading_pattern_inequality}
Consider a feasible ELP of an edge $\tau_i \rightarrow \tau_j$ and two jobs $J_{j,q_u}$ and $J_{j,q_w}$ where $q_u < q_w$, their last-reading jobs specified by the ELP, $J_{i, \overleftarrow{q_u}}$ and $J_{i, \overleftarrow{q_w}}$, satisfy:
    $\overleftarrow{q_u} \leq \overleftarrow{q_w}$.
\end{lemma}
\begin{proof}
    If $\overleftarrow{q_u} > \overleftarrow{q_w}$, then $J_{i,\overleftarrow{q_w}}$ cannot become $J_{j,q_w}$'s last reading job because $J_{j,q_w}$ could read from $J_{i,\overleftarrow{q_u}}$ since $wr_{J_{i,\overleftarrow{q_w}}} < wr_{J_{i,\overleftarrow{q_u}}} < re_{J_{j,q_u}} <re_{J_{j,q_w}}$. This causes a contradiction.
\end{proof}

\begin{lemma}
\label{lemma_two_task_Reacting_pattern_inequality}
Consider a feasible EFP of an edge $\tau_i \rightarrow \tau_j$ and two jobs $J_{i,q_u}$ and $J_{i,q_w}$
where $q_u < q_w$, their first-reacting jobs specified by the EFP, $J_{j, \overrightarrow{q_u}}$ and $J_{j, \overrightarrow{q_w}}$, satisfy $\overrightarrow{q_u} \leq \overrightarrow{q_w}$.
\end{lemma}
\begin{proof}
    The proof is similar to that of Lemma~\ref{lemma_two_task_Reading_pattern_inequality}.
\end{proof}

\begin{Example}
\label{example_TTCP}  
Table~\ref{DA_TCP_example} shows all the feasible ELPs in example~\ref{Example_setup} under different situations (since one job could read from different jobs). The mathematical expressions show the intersection of all the possible linear inequalities, following Equation~\eqref{eq_ecp_fr}. The linear inequalities are merged following \defref{elp inequalities}.
\end{Example}

\begin{table}[]
\centering
\caption{Feasible edge last-reading patterns in Example~\ref{Example_setup}}
\label{DA_TCP_example}
\resizebox{0.45\textwidth}{!}{%
\begin{tabular}{@{}llll@{}}
\toprule
Edge &
  \multicolumn{1}{c}{\begin{tabular}[c]{@{}c@{}}ELP Pattern\\ Symbol\end{tabular}} &
  \multicolumn{1}{c}{\begin{tabular}[c]{@{}c@{}}Last-reading\\ job pairs\end{tabular}} &
  \multicolumn{1}{c}{\begin{tabular}[c]{@{}c@{}}Mathematical\\ expression\end{tabular}} \\ \midrule
$E_0$ & $\boldsymbol{ELP}_{E_0}(0)$ & $J(0,2) \rightarrow J(1,0)$                                                                         & $D_0+10  \leq O_1 < D_0  +15$  \\ \midrule
$E_0$ & $\boldsymbol{ELP}_{E_0}(1)$ & $J(0,1) \rightarrow J(1,0)$                                                                         & $D_0+5  \leq O_1 < D_0+10  $   \\ \midrule
$E_0$ & $\boldsymbol{ELP}_{E_0}(2)$ & $J(0,0) \rightarrow J(1,0)$                                                                         & $D_0  \leq O_1 < D_0+5  $      \\ \midrule
$E_0$ & $\boldsymbol{ELP}_{E_0}(3)$ & $J(0,-1) \rightarrow J(1,0)$                                                                        & $D_0 -5 \leq O_1 < D_0  $      \\ \midrule
$E_1$ & $\boldsymbol{ELP}_{E_1}(0)$ & \begin{tabular}[c]{@{}l@{}}$J(1,0) \rightarrow J(2,0)$\\ $J(1,0) \rightarrow J(2,1)$\end{tabular}   & $D_1  \leq O_2 < D_1 +3 $      \\ \midrule
$E_1$ & $\boldsymbol{ELP}_{E_1}(1)$ & \begin{tabular}[c]{@{}l@{}}$J(1,-1) \rightarrow J(2,0)$\\ $J(1,0) \rightarrow J(2,1)$\end{tabular}  & $D_1 -10  \leq O_2 < D_1  $    \\ \midrule
$E_1$ & $\boldsymbol{ELP}_{E_1}(2)$ & \begin{tabular}[c]{@{}l@{}}$J(1,-1) \rightarrow J(2,0)$\\ $J(1,-1) \rightarrow J(2,1)$\end{tabular} & $D_1 -20  \leq O_2 < D_1 -10 $ \\ \midrule
$E_2$ & $\boldsymbol{ELP}_{E_2}(0)$ & \begin{tabular}[c]{@{}l@{}}$J(3,0) \rightarrow J(1,0)$\\ $J(3,0) \rightarrow J(1,1)$\end{tabular}   & $D_3  \leq O_1 < D_3 +8  $     \\ \midrule
$E_2$ & $\boldsymbol{ELP}_{E_2}(1)$ & \begin{tabular}[c]{@{}l@{}}$J(3,-1) \rightarrow J(1,0)$\\ $J(3,0) \rightarrow J(1,1)$\end{tabular}  & $D_3 -20 \leq O_1 < D_3  $     \\ \midrule
$E_2$ & $\boldsymbol{ELP}_{E_2}(2)$ & \begin{tabular}[c]{@{}l@{}}$J(3,-1) \rightarrow J(1,0)$\\ $J(3,-1) \rightarrow J(1,1)$\end{tabular} & $D_3 -40 \leq O_1 < D_3 -20 $  \\ \bottomrule
\end{tabular}%
}
\end{table}

\subsection{Graph Communication Pattern (GCP)}

Given a DAG, we introduce the concept of \textit{Graph Communication Pattern} (GCP) to comprehensively describe the reading and writing relationships among all the tasks in the system.
\begin{definition}[Graph last-reading pattern, GLP]
\label{def: glp}
    Consider a DAG $\boldsymbol{G}=(\boldsymbol{\tau}, \boldsymbol{E})$, 
    a graph last-reading pattern $\boldsymbol{GLP}$ is a set of edge last-reading patterns, where for each edge $E_k \in \boldsymbol{E}$, $\boldsymbol{GLP}$ contains one ELP, denoted as $\boldsymbol{GLP}(E_k)$.
\end{definition}
\begin{definition}[Partial GLP]
\label{def: partial glp}
    A partial GLP $\boldsymbol{pGLP}$ is a GLP where some or zero edges' ELPs are not included. 
\end{definition}
A GLP is complete if the ELPs of all the edges are added; A complete GLP is considered a partial GLP. A GLP is empty if no ELPs are added. The concepts of Graph First-reading Pattern (GFP) and partial GFP are defined similarly. Both GLP and GFP are considered GCP.

\begin{definition}[Feasibility]
\label{def: feasibility}
    A GLP (GFP) is feasible if there is a set of virtual offset and deadline variables that satisfy the linear inequalities of all the ELPs (EFPs). 
\end{definition}
Evaluating the feasibility of a GLP (GFP) is equivalent to evaluating the feasibility of a linear programming (LP) problem, which can be efficiently solved to optimality.

\begin{figure}[t!]
\centering
\includegraphics[width=1.0\columnwidth]{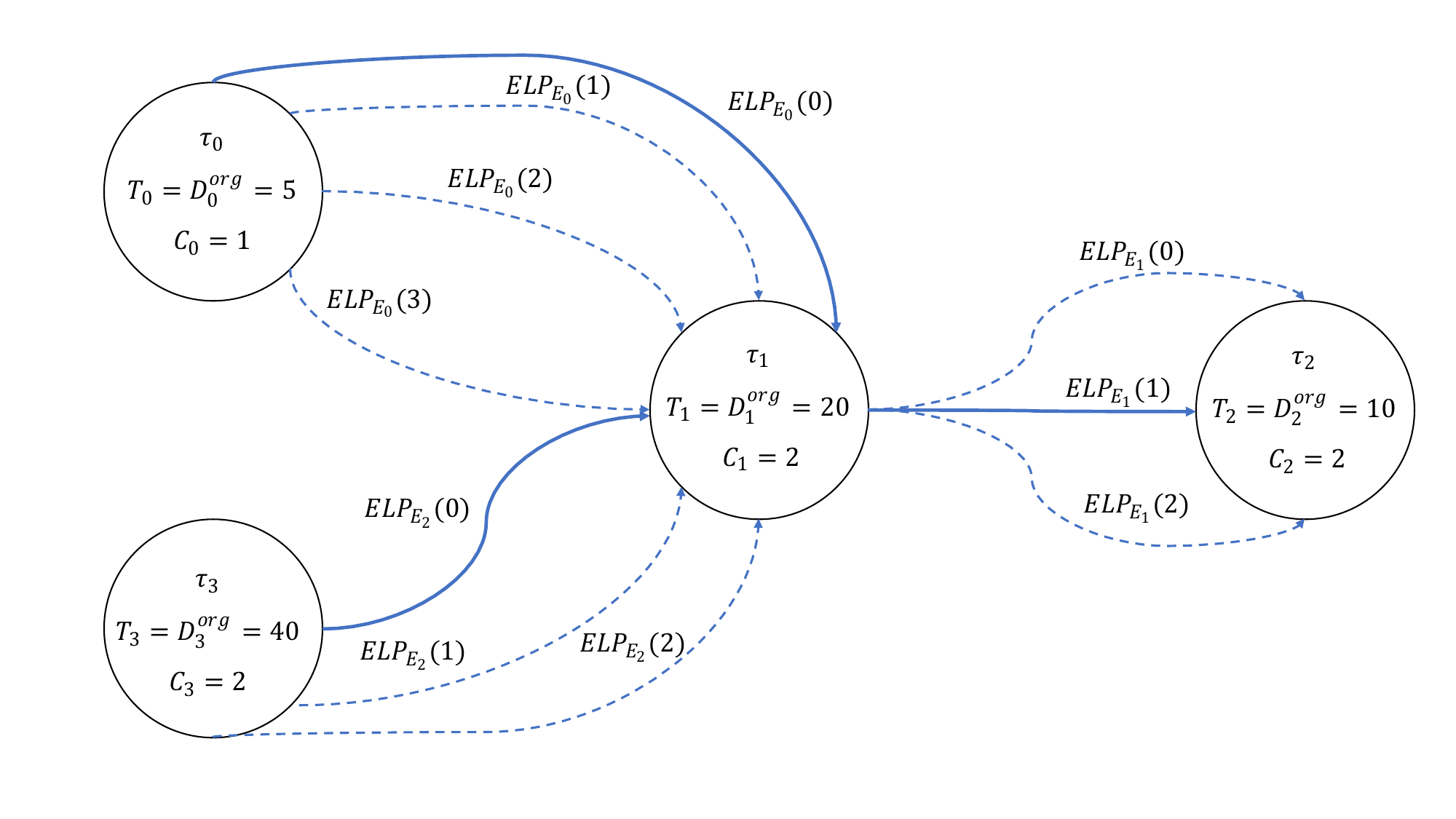}
\caption{
A multi-graph with the edge last-reading patterns (ELPs) in Example~\ref{Example_setup}. 
A complete graph last-reading pattern (GLP) comprises three ELPs from the three edges. There are $4\times 3 \times 3 = 36$ possible GLP combinations.
}
\label{fig_example_dag_multi_graph}
\end{figure}

\subsection{Evaluating Graph Communication Pattern}
\label{section_evaluate_gcp}

\begin{definition}[Graph last-reading pattern evaluation]
\label{def: glp evaluate}
    A GLP's evaluation, denoted $\hat{\mathcal{F}}_{\boldsymbol{GLP}}$, is the optimal objective function values obtained by solving an optimization problem (e.g., \eqref{obj_overall} or~\eqref{obj_overall_sf}) with extra schedulability constraints and linear constraints specified by $\boldsymbol{GLP}$.
\end{definition}
For example, if the objective function is given by~\eqref{obj_overall}, then:
\begin{align}
 \label{eq_def_f_g}
\hat{\mathcal{F}}_{\boldsymbol{GLP}} &= 
    \min_{\boldsymbol{O},\boldsymbol{D}}
    \sum_{\mathcal{C} \in \boldsymbol{\mathcal{C}}}
     \max \boldsymbol{\mathcal{F}}_{\mathcal{C}}(\boldsymbol{O},\boldsymbol{D})
\end{align}
\begin{equation}
    \forall \tau_i \in \boldsymbol{\tau}, 0 \leq O_i, \ O_i+R_i \leq D_i,\ D_i \leq D_i^{org}
\label{schedulability_constraint_gcp_eval}
\end{equation}
\begin{equation}
\forall E_k \in \boldsymbol{E},\  \mathcal{I}_{\boldsymbol{GLP}(E_k)}
\end{equation}
\noindent where $\mathcal{I}_{\boldsymbol{GLP}(E_k)}$ denotes the linear inequality constraint~\eqref{eq_ecp_fr} of the edge last-reading pattern $\boldsymbol{GLP}(E_k)$.
GFP's evaluation is defined similarly.

\begin{Example}
\label{example_evaluate_gcp}
    Continue with Example~\ref{Example_setup} and we will evaluate $\glp =\{\elp_{E_0}(0)$,
    $\boldsymbol{ELP}_{E_1}(1)$,
   $\boldsymbol{ELP}_{E_2}(0)\}$ with only \textbf{one} cause-effect chains $\mathcal{C}_0$. Consider a hyper-period of 20 (without considering $\tau_3$), $\tau_2$ has two jobs: $J_{2,0}$ and $J_{2,1}$. All the immediate backward job chains have been decided by $\boldsymbol{GLP}$: 
    $J_{0,-2} \rightarrow  J_{1,-1}  \rightarrow J_{2,0}$,
    $J_{0,2} \rightarrow  J_{1,0}  \rightarrow J_{2,1}$.
    The objective function becomes:
    \begin{align}\min_{\boldsymbol{O},\boldsymbol{D}} \ &
        (\max (wr_{J_{2,0}}-re_{J_{0,-2}}, \ wr_{J_{2,1}}-re_{J_{0,2}} ))
        \label{obj_example}
    \end{align}
    where the reading and writing times are given by Definition~\ref{def_virtual_offset} and~\ref{def_virtual_deadline}. The objective~\eqref{obj_example} can be converted into linear functions with extra artificial variables~\cite{Chen2010AppliedIP}: $a=wr_{J_{2,0}}-re_{J_{0,-2}}$, $b=wr_{J_{2,1}}-re_{J_{0,2}}$, $c=\max (a, b)$:
     \begin{align}
        \min_{\boldsymbol{O},\boldsymbol{D}} \ & c \\
        c \geq a,\  c &  \geq b
    \end{align}
    The ELP constraints of $\boldsymbol{GLP}$ are adopted from Table~\ref{DA_TCP_example}:
    \begin{align}
        D_0+10  &\leq O_1 < D_0  +15 \\
        D_1 -10  &\leq O_2 < D_1  \\
        D_3  &\leq O_1 < D_3 +8
    \end{align}
    The schedulability constraints~\eqref{schedulability_constraint_gcp_eval} remain the same. The optimal solutions of the LP above is $\boldsymbol{O}=\{0,11,6,0\}$, $\boldsymbol{D}=\{1,16,9,40\}$, 
$\hat{\mathcal{F}}_{\boldsymbol{GLP}}=19$.
\end{Example}

\begin{lemma}
\label{lemma_min_max_lp}
    Consider a min-max optimization problem 
    \begin{equation}
        \min_{\textbf{x}} \max (f_1(\textbf{x}), ..., f_n(\textbf{x}))
    \end{equation}
    where $f_i(\textbf{x})$ are linear functions. The problem above can be transformed into a linear programming (LP) problem by introducing extra continuous variables for each $f_i(\textbf{x})$.
\end{lemma}
Lemma~\ref{lemma_min_max_lp} shows that evaluating GLPs is usually equivalent to LP.
An example of how to perform the transformation is shown in Example~\ref{example_evaluate_gcp}. 

\begin{theorem}
\label{theorem_glp_lp}
    Evaluating a GLP when the objective function is data age or time disparity is equivalent to solving an LP.
\end{theorem}
\begin{proof}
Since the schedulability constraints~\eqref{schedulability_constraint_gcp_eval} and all the ELP constraints (\defref{elp inequalities}) are linear, we will focus on the objective function in the proof.
Given a GLP, all the immediate backward job chains are uniquely decided. Therefore, the data age of each immediate backward job chain is a linear function of $\textbf{O}$ and $\textbf{D}$ (Definition~\ref{def_length_job_chain}). 

As for time disparity, equation~\eqref{time_disparity_def} can be transformed into linear functions by introducing two continuous artificial variables: $a=\max_{J \in \boldsymbol{J}^{lr,s} }wr_J$, $b=\min_{J \in \boldsymbol{J}^{lr,s}} wr_J$ with extra linear constraints:
    \begin{equation}
        \forall J \in \boldsymbol{J}^{lr,s}, a \geq wr_J, \ b \leq wr_J
    \end{equation}
    
    The theorem is proved after combining with Lemma~\ref{lemma_min_max_lp}.
\end{proof}

\begin{theorem}
    Evaluating a GFP when the objective function is reaction time is equivalent to solving an LP.
\end{theorem}
\begin{proof}
    Skipped because it is similar to Theorem~\ref{theorem_glp_lp}.
\end{proof}

\begin{observation}
\label{lemma_td_jitter_linearizable}
    Evaluating a GLP when the objective function is a time disparity jitter could be equivalent to mixed integer linear programming by following the methods explained in Chapter 3 in Chen \etal~\cite{Chen2010AppliedIP}. How to do it exactly and perform efficient GLP optimization are left as future work.
\end{observation}

\section{Two-stage Optimization}
\label{section_optimization}
Finding the optimal virtual offset and virtual deadline for problem~\eqref{obj_overall} or \eqref{obj_overall_sf} can be solved in two stages: 
\begin{itemize}[leftmargin=*]
    \item Enumerate all the possible GLPs (GFPs);
    \item Evaluate (\defref{glp evaluate}) and select the GLP (GFP) with the minimum objective function value;
\end{itemize}
The optimal virtual offset and virtual deadline are obtained when evaluating the optimal GLPs (GFPs). 

The two-stage method is beneficial because numerous infeasible and non-optimal GLPs can be quickly skipped, as will be introduced next.
Since all the definitions and operations of GLP can be symmetrically applied to GFP, we will primarily use GLP to illustrate the algorithms and mention only the necessary changes when optimizing GFPs. 

\subsection{Finding Optimal GCP, Enumeration Method}
\label{section_opt_bf}
A simple way to find the optimal GLP is to enumerate all the possible ELP combinations. 
Since each edge $E_k$ has multiple ELPs to select, the enumeration process can be modeled as a multi-graph~\cite{enwiki_multi_graph}. 
\begin{Example}
\label{example_gcp_bf}
    In Fig.~\ref{fig_example_dag_multi_graph}, there are $36$ GLPs to enumerate.
\end{Example}

\subsection{Finding Optimal GCP, Backtracking}
\label{section_opt_skip}
Since evaluating GLP's feasibility is equivalent to evaluating linear programming (LP)'s feasibility (\defref{feasibility}), the enumeration method can be expedited using a backtracking algorithm. 
During enumeration, feasibility is assessed each time a new ELP is added to a partial GLP \Pite{}. If \Pite{} is infeasible, then no other ELPs will be added to it.

Compared with the enumeration method, the backtracking method avoids evaluating infeasible GLPs while maintaining optimality, resulting in significantly faster speed. 
\begin{Example}
\label{example_skip_infeasible_gcp}
    Continue with Example~\ref{example_gcp_bf}, 10 out of the 36 possible GLPs are not feasible and will be skipped, such as $\{ \boldsymbol{ELP}_{E_0}(0)$, 
    $\boldsymbol{ELP}_{E_1}(1)$, 
    $\boldsymbol{ELP}_{E_2}(0)\}$.
\end{Example}

\begin{theorem}
    The enumeration and the backtracking method above guarantee finding the optimal solutions if not time-out.
\end{theorem}
\begin{proof}
    Since both methods evaluate all the feasible GLPs (GFPs), the solutions are guaranteed to be optimal.
\end{proof}

\section{DART Symbolic Optimization}
\label{section_dart_optimization}
This section proposes an efficient backtracking algorithm based on symbolic operations when minimizing data age or reaction time (DART). 
During iterations, many feasible, but \textit{possibly non-optimal} GLPs (GFPs) will be skipped based on new theorems. The algorithm guarantees that the performance of the final result falls within a small bound of the optimal solution (\ie bounded suboptimality).
The algorithm is motivated by an example shown in Fig.~\ref{let_skip_example_fig}. 

\begin{figure}[t!]
\centering
\includegraphics[width=1.0\columnwidth]{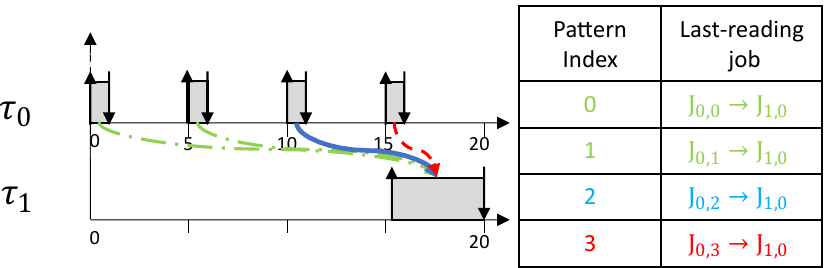}
\caption{Communication pattern comparison. Consider a simple DAG with only one cause-effect chain $\tau_0 \rightarrow \tau_1$. Furthermore, to illustrate the idea more easily, we assume that $\tau_1$'s virtual deadline is \textit{locked} at 20. In this case, the optimal edge last-reading pattern (ELP) is \elpmore{0}{2}: \elpmore{0}{3} is infeasible; \elpmore{0}{2} has shorter end-to-end latency than \elpmore{0}{1} and \elpmore{0}{0} (\elpmore{0}{0} and \elpmore{0}{1} require moving $\tau_1$'s virtual offset forward).
Notice that we can find the optimal ELP based on the pattern index with a feasibility check (a larger index implies a shorter immediate backward job chain).
} 
\label{let_skip_example_fig}
\end{figure}

\subsection{Comparing Edge Communication Patterns}
\label{section_compare_ecp}
\begin{definition}[ELP Comparison, smaller, $\preccurlyeq $]
\label{def: elp comp}
    Consider two edge last-reading patterns \elpmore{k}{u} and $\boldsymbol{ELP}_{E_k}(w)$ for an edge $E_k=\tau_i \rightarrow \tau_j$. For any job $J_{j,q_j}$, if the job index of its last-reading job in $\boldsymbol{ELP}_{E_k}(u)$ is always less than or equal to those in $\boldsymbol{ELP}_{E_k}(w)$, then $\boldsymbol{ELP}_{E_k}(u) \preccurlyeq \boldsymbol{ELP}_{E_k}(w)$. 
\end{definition}
\begin{definition}[GLP Comparison, smaller, $\preccurlyeq $]
\label{def: glp comp}
    If, for each edge $E_k$, the ELPs of two GLPs satisfy $\boldsymbol{GLP}_{u}(E_k) \preccurlyeq \boldsymbol{GLP}_{w}(E_k)$, then $\boldsymbol{GLP}_{u} 
\preccurlyeq \boldsymbol{GLP}_{w}$.
\end{definition}

Similarly, we can define the comparison for EFP and GFP, except we compare the job index of first-reacting jobs.

\begin{Example}
    Consider the edge last-reading patterns in Table~\ref{DA_TCP_example}. We have $ \boldsymbol{ELP}_{E_1}(2)$ $\preccurlyeq$ $\boldsymbol{ELP}_{E_1}(1)$ $\preccurlyeq $ $\boldsymbol{ELP}_{E_1}(0)$. 
    However, consider two imaginary ELPs:
    \begin{itemize}
    \item $\boldsymbol{ELP}_{E_1}(3): 
   \{J(1,0)\rightarrow J(2,0), 
    J(1,2) \rightarrow J(2,1)\}$    
    \item $\boldsymbol{ELP}_{E_1}(4):\{J(1,1) \rightarrow J(2,0), 
    J(1,1) \rightarrow J(2,1)\}$
    \end{itemize}
    then neither $\boldsymbol{ELP}_{E_1}(3) \preccurlyeq \boldsymbol{ELP}_{E_1}(4)$ nor $\boldsymbol{ELP}_{E_1}(4) \preccurlyeq \boldsymbol{ELP}_{E_1}(3)$ holds. (Note $\boldsymbol{ELP}_{E_1}(3)$ and $\boldsymbol{ELP}_{E_1}(4)$ are only used for illustration and do not exist in Example~\ref{Example_setup}).
\end{Example}

\subsection{Identifying Non-optimal GLPs (GFPs)}
\begin{theorem}
\label{theorem_DA_chain}
Consider a cause effect chain $\mathcal{C}:\tau_0 \rightarrow ... \rightarrow  \tau_n$ and two feasible graph last-reading patterns $\boldsymbol{GLP}_{u} \preccurlyeq \boldsymbol{GLP}_{w}$.
For any job $J_{n, q_n}$, denote $\tau_0$'s job in $J_{n, q_n}$'s immediate backward job chain in $\boldsymbol{GLP}_{u}$ and $\boldsymbol{GLP}_{w}$ as $J_{0, q^u_0}$ and $J_{0, q^w_0}$, then $q^u_0\leq q^w_0$.
\end{theorem}
\begin{proof}
We prove the theorem by induction. 
Following \defref{glp comp}, the theorem holds when $\mathcal{C}$ has only two tasks.
Next, we assume the theorem holds for the cause-effect chain $\tau_i \rightarrow ... \rightarrow \tau_n$, and prove that it holds for $\tau_{i-1} \rightarrow ... \rightarrow \tau_n$.

Let's denote a job $J_{n,q_n}$'s last reading job at task $\tau_{i}$ under $\boldsymbol{GLP}_{u}$ and $\boldsymbol{GLP}_{w}$ as $J_{i, q^u_i }$ and $J_{i, q^w_i }$. We know
$q^u_i \leq q^w_i $ based on the induction assumption. 
Next, at task $\tau_{i-1}$, let's denote 
the last-reading job of $J_{i, q^u_i}$ under $\boldsymbol{GLP}_{u}$ as $J_{i-1, \overleftarrow{q}^u_i}$,
the last-reading job of $J_{i, q^w_i}$ under $\boldsymbol{GLP}_{u}$ as $J_{i-1, \overleftarrow{q}^{wu}_i}$,
the last-reading job of $J_{i, q^w_i}$ under $\boldsymbol{GLP}_{w}$ as $J_{i-1, \overleftarrow{q}^w_i}$. 
Then we have $\overleftarrow{q}^u_i \leq \overleftarrow{q}^{wu}_i $ from Lemma~\ref{lemma_two_task_Reading_pattern_inequality}, and $\overleftarrow{q}^{wu}_i \leq \overleftarrow{q}^w_i $ from the assumption that $\boldsymbol{GLP}_{u} \preccurlyeq \boldsymbol{GLP}_{w}$ and \defref{elp comp}. Therefore, $\overleftarrow{q}^u_i \leq \overleftarrow{q}^w_i$, which proves the induction.
\end{proof}

\begin{theorem}
\label{theorem_RT_chain}
Consider a cause-effect chain $\mathcal{C}:\tau_0 \rightarrow ... \rightarrow  \tau_n$, and two feasible graph first-reacting patterns $\boldsymbol{GFP}_{u} 
\preccurlyeq \boldsymbol{GFP}_{w}$.
Now consider a job $J_{0, q_0}$ of $\tau_0$, 
denote $\tau_n$'s job in $J_{0, q_0}$'s immediate forward job chain in $\boldsymbol{GFP}_{u}$ and $\boldsymbol{GFP}_{w}$ as $J_{n, q_n^u}$ and $J_{n, q_n^w}$, then $q_n^u \leq q_n^w$.
\end{theorem}
\begin{proof}
    Skipped because it is similar to Theorem~\ref{theorem_DA_chain}.
\end{proof}

\begin{lemma}
\label{theorem_compare_mn_bound}
Continue with Theorem~\ref{theorem_DA_chain}, denote the immediate backward job chain that terminates at $J_{n,q_n}$ in $\boldsymbol{GLP}_u$ and $\boldsymbol{GLP}_w$ as $\mathcal{C}^J_u$ and $\mathcal{C}^J_w$. Then we have $\text{Len}(\mathcal{C}^J_w) -\text{Len}(\mathcal{C}^J_u) \leq B_{\mathcal{C}}$, where $B_{\mathcal{C}}\leq T_n + T_0$.
\end{lemma}
\begin{proof}
Let's first introduce the related notations:
\begin{align}
  \text{Len}(\mathcal{C}^J_w) 
  = &wr^w_{J_{n,q_n}} - re^w_{J_{0,q_0^w}} \\ 
    \text{Len}(\mathcal{C}^J_u) =
    & wr^u_{J_{n,q_n}} - re^u_{J_{0,q_0^u}}\\
    \text{Len}(\mathcal{C}^J_w) -\text{Len}(\mathcal{C}^J_u) 
    = &(wr^w_{J_{n,q_n}}-wr^u_{J_{n,q_n}}) + \\ 
    &-(re^w_{J_{0,q_0^w}} - re^u_{J_{0,q_0^u}})
\end{align}
    Since $\boldsymbol{GLP}_u$ and $\boldsymbol{GLP}_w$ are feasible and $q_0^u \leq q_0^w$, we know $re^u_{J_{0, q_0^u}} < re^w_{J_{0, q_0^w}}$ if $q_0^u < q_0^w$, 
    and $|re^w_{J_{0, q_0^w}} - re^u_{J_{0, q_0^u}} | \leq T_0$ if $q_0^u = q_0^w$. Similarly, $wr^w_{J_{n,q_n}}-wr^u_{J_{n,q_n}} \leq T_n$. Therefore, $\text{Len}(\mathcal{C}^J_w) -\text{Len}(\mathcal{C}^J_u) \leq T_n+T_0$.
\end{proof}

\begin{theorem}
\label{theorem_bound_single_chain}
Continue with Theorem~\ref{theorem_DA_chain},  the worst-case data age of $\boldsymbol{GLP}_u$ cannot be smaller than $\boldsymbol{GLP}_w$ by more than $B_{\mathcal{C}}$.
\end{theorem}
\begin{proof}
    Let's use two vectors $\textbf{U}$ and $\textbf{W}$ to denote the length of the immediate backward job chains in $\boldsymbol{GLP}_u$ and $\boldsymbol{GLP}_w$, respectively. For any job $J_{n, q_n}$, let's use $\textbf{U}_{q_n}$ and $\textbf{W}_{q_n}$ to denote the length of the job chains that terminate at $J_{n, q_n}$. Following Lemma~\ref{theorem_DA_chain}, we have $\textbf{W}_{q_n} - \textbf{U}_{q_n} \leq B_{\mathcal{C}}$.
    Next, let's denote the index of the longest job chain in $\textbf{U}$ and $\textbf{W}$ as $k_u$ and $k_w$. Then we have $\Delta = \textbf{U}_{k_u} - \textbf{U}_{k_w} \geq 0$. Therefore,
    \begin{equation}
        \textbf{W}_{k_w}-\textbf{U}_{k_u} = \textbf{W}_{k_w}-(\textbf{U}_{k_w} + \Delta ) \leq  B_{\mathcal{C}} - \Delta \leq  B_{\mathcal{C}}
    \end{equation}
\end{proof}
\begin{theorem}
\label{theorem_bound_multi_chain}
    Consider a DAG with multiple cause-effect chains, then Theorem~\ref{theorem_DA_chain} and~\ref{theorem_bound_single_chain} hold for each cause-effect chain.
\end{theorem}
\begin{proof}
    Because the derivation of the performance bound of each cause-effect chain does not interfere with each other.
\end{proof}

\begin{Example}
In Example~\ref{Example_setup}, consider two graph last-reading patterns $\boldsymbol{GLP}_u = \{ \boldsymbol{ELP}_{E_0}(0),
    \boldsymbol{ELP}_{E_1}(1),
    \boldsymbol{ELP}_{E_2}(1)\}$ and  
    $\boldsymbol{GLP}_w = \{\boldsymbol{ELP}_{E_0}(0),
    \boldsymbol{ELP}_{E_1}(1),
    \boldsymbol{ELP}_{E_2}(0)\}$. 
    We have $\boldsymbol{GLP}_u \preccurlyeq \boldsymbol{GLP}_w$ following \defref{glp comp}. 
    Hence, based on Theorem~\ref{theorem_bound_single_chain},
    \glpevalOneindex{w} (\ie the data age after optimization) is guaranteed to be smaller or no larger than $B_{\mathcal{C}}$ (Lemma~\ref{theorem_compare_mn_bound}) when compared to \glpevalOneindex{u}. 
\end{Example}

\subsection{Skip Partial Graph Communication Patterns}
Next, we generalize Theorem~\ref{theorem_DA_chain} to partial graph communicating patterns.
When evaluating a partial GLP's data age, we only consider the partial cause-effect chains formulated by the edges contained in the partial GLP. If the partial cause-effect chain does not include the source task of the complete chain, we consider its data age to be 0.
\begin{Example}
Consider a cause-effect chain 
$\tau_0 \rightarrow \tau_1 \rightarrow \tau_2$ 
and two partial GLPs $\boldsymbol{pGLP}_1=\{ \boldsymbol{ELP}_{E_0}(1) \}$ and $\boldsymbol{pGLP}_2=\{\boldsymbol{ELP}_{E_1}(2) \}$,
then the data age evaluation of $\boldsymbol{pGLP}_1$ only considers the cause-effect chain $\tau_0 \rightarrow \tau_1$. The data age of $\boldsymbol{pGLP}_2$ is 0 because it does not contain the source task $\tau_0$.
\end{Example}

\begin{theorem}
\label{theorem_partial_dart_chain}
    Theorem~\ref{theorem_DA_chain} applies to partial graph communication patterns.
\end{theorem}
\begin{proof}
    Skipped due to being similar to Theorem~\ref{theorem_DA_chain}.
\end{proof}

\begin{definition}[Contain, $\subset$]
\label{def: contain}
    Consider two partial GLPs $\boldsymbol{pGLP}_u$ and $\boldsymbol{pGLP}_w$. We say $\boldsymbol{pGLP}_w$ contains $\boldsymbol{pGLP}_u$, denoted as $\boldsymbol{pGLP}_u \subset \boldsymbol{pGLP}_w$, if:
    \begin{equation*}
        \forall \boldsymbol{pGLP}_u(E_k) \in \boldsymbol{pGLP}_u,  \ \boldsymbol{pGLP}_u(E_k)=\boldsymbol{pGLP}_w(E_k)
    \end{equation*}
\end{definition}
\begin{Example}
    Consider three partial graph last-reading patterns $\boldsymbol{pGLP}_u = \{ \boldsymbol{ELP}_{E_0}(0)\}$, $\boldsymbol{pGLP}_w = \{ \boldsymbol{ELP}_{E_0}(0),$
    $\boldsymbol{ELP}_{E_1}(1)\}$ and $\boldsymbol{pGLP}_k = \{ \boldsymbol{ELP}_{E_0}(1)\}$. Then $\boldsymbol{pGLP}_w$ contains $\boldsymbol{pGLP}_u$, but does not contain $\boldsymbol{pGLP}_k$.
\end{Example}

\begin{theorem}
\label{incomplete_chain_worse_theorem}
    Consider two sub-chains $\mathcal{C}_{1}=\tau_0 \rightarrow ... \rightarrow \tau_k$ and $\mathcal{C}_{2}=\tau_{k+1} \rightarrow ... \rightarrow \tau_n$ of a cause-effect chain $\mathcal{C}=\mathcal{C}_{1} \rightarrow \mathcal{C}_{2}$. 
    Given the same task set schedule, the worst-case data age of $\mathcal{C}_{1}$ is no larger than $\mathcal{C}$.
\end{theorem}
\begin{proof}
    This theorem can be proved following the compositional theorem and its generalization in G{\"u}nzel~\etal~\cite{Gnzel2021TimingAO}.
\end{proof}

\begin{theorem}
\label{theorem_partial_gcp}
Consider two partial GLPs $\boldsymbol{pGLP}_u \subset \boldsymbol{pGLP}_w$, the maximum data age evaluated in $\boldsymbol{pGLP}_w$ (following \defref{glp evaluate}) is no less than $\boldsymbol{pGLP}_u$.
\end{theorem}
\begin{proof}
    Direct results of Theorem~\ref{incomplete_chain_worse_theorem} because the cause-effect chains in $\boldsymbol{pGLP}_w$ have more tasks.
\end{proof}

\subsection{DART-specific Symbolic Optimization Algorithm}
\label{section_dart_symbopt}
The theorems above can be summarized as follows:
\begin{corollary}
\label{theormSummary}
Consider an optimization problem that minimizes data age.
A partial GLP $\boldsymbol{pGLP}_i$ and all the other partial GLPs that contain $\boldsymbol{pGLP}_i$ can be skipped if one of the evaluated partial GLP $\boldsymbol{pGLP}_j$ satisfies (\cref{theorem_partial_dart_chain,incomplete_chain_worse_theorem,theorem_partial_gcp}): 
\begin{itemize}
    \item $\boldsymbol{pGLP}_j$'s evaluation is smaller than $\boldsymbol{pGLP}_i$;
    \item $\boldsymbol{pGLP}_i$ is smaller than $\boldsymbol{pGLP}_j$ (\defref{glp comp});
\end{itemize} 
\end{corollary}

\noindent \textbf{ELP Enumeration Order.} 
ELPs need to be ordered appropriately to fully utilize the symbolic operations.
When selecting edges to add to GLPs, select the edges that are closer to the source tasks first. 
After selecting an edge, first add its largest ELP (follow \defref{elp comp}) into the GLP when minimizing data age (smallest EFP first for reaction time optimization). 

\noindent \textbf{Pseudocode.} The pseudocode for data age minimization problem is shown in Algorithm~\ref{alg_opt_skipped}.
The vector $\textbf{best\_yet\_obj}$ denotes the worst-case data age of each cause-effect chain evaluated by the best-known GLP. 
In \alglineref{GetBestPossibleObj}, the implementation of the function $\textbf{GetBestPossibleObj}$ follows \defref{glp evaluate} except only considering the sub-chains contained in $\textbf{P}_{ite}$.
The $\geq$ operator in \alglineref{compareBestYetObj} requires element-wise greater than to hold.
In \alglineref{worsePatterns}, $\textbf{worsePatterns}$ is implemented as a First-In-First-Out queue, because recently visited GLPs have more common ELPs with the next GLP, increasing the likelihood of triggering the skip conditions in Corollary~\ref{theormSummary}. 
The capacity size is experimentally set to 50: a larger capacity incurs more computation overhead when a communication pattern cannot be skipped, while a smaller size implies fewer skipped patterns.
In \alglineref{RemoveWorseELP}, the function \textbf{RemoveWorseELP} removes an unvisited ELP from $\textbf{ELP}\_\textbf{Set}$ if adding the ELP to $\textbf{P}_{ite}$ will make $\textbf{P}_{ite}$ perform worse (evaluated by Corollary~\ref{theormSummary}) than at least one partial GLP in $\textbf{worse}\_\textbf{patterns}$. This function significantly speeds up symbolic optimization by allowing the skipping of partial GLPs, eliminating the need to visit and evaluate all partial and complete GLPs that contained the skipped partial GLP.

\begin{algorithm}[ht!]
\SetAlgoLined
\SetKwInOut{a}{b}
\caption{\textbf{GLPSymbOpt}}
\label{alg_opt_skipped}

\KwIn{ DAG $\textbf{G} = (\boldsymbol{\tau}, \textbf{E})$, GLP $\textbf{P}_{ite}=\{ \}$, $\textbf{best\_yet\_obj}=\{Inf,...,Inf\}$, $\textbf{worsePatterns}=\{ \}$, unvisited$\_$Edge $E$=$\textbf{E}(0)$}
\KwOut{$\textbf{best\_yet\_obj}$}
\If{$\textbf{P}_{ite}$ is complete}{
Solve problem~\eqref{eq_def_f_g} and Update $\textbf{best\_yet\_obj}$\\
\Return 
}

$\textbf{ELP$\_$Set} = \textbf{GetAllELP}(E)$\\

\While{$\textbf{ELP$\_$Set} \text{ is not empty}$}
{
    $ELP=\textbf{ELP$\_$Set.}\textbf{Pop\_front()}$\\
    $\textbf{P}_{ite}\textbf{.}\textbf{Insert}(ELP)$ \\
    \If{$\textbf{P}_{ite}$ is feasible 
    } 
    { 
         $\textbf{obj}= \textbf{GetBestPossibleObj}(\textbf{P}_{ite}) $ \label{line: GetBestPossibleObj} \\
        \eIf{$\textbf{obj} \ \boldsymbol{\geq} \textbf{best\_yet\_obj} $ \label{line: compareBestYetObj} } 
        {
            $\textbf{worsePatterns}\textbf{.}\textbf{push\_back}(\textbf{P}_{ite})$ \label{line: worsePatterns} \\
            \If{$\textbf{worsePatterns}\textbf{.}\textbf{size()} > 50$}{           \textbf{worsePatterns}\textbf{.}\textbf{pop$\_$front()}
            }
                
        }
        {
            $\textbf{GLPSymbOpt}(\textbf{G}, \textbf{P}_{ite}, \textbf{best\_yet\_obj}$, $\textbf{worsePatterns}, \textbf{NextUnvisitedEdge}(\textbf{P}_{ite}))$
        }
    
        $\textbf{ELP$\_$Set}\textbf{.}\textbf{RemoveWorseELP}(\textbf{worsePatterns})$  \label{line: RemoveWorseELP}  
    }
    $\textbf{P}_{ite}\textbf{.}\textbf{Erase}(ELP)$ \\
}
\end{algorithm}
\begin{Example}
    In Example~\ref{Example_setup}, the edge iteration order is $E_0 \rightarrow E_2 \rightarrow E_1$ because $E_0$ and $E_2$ contain the source tasks in cause-effect chains. 
    The first GLP to evaluate is $\boldsymbol{GLP}_0=$ $\{\boldsymbol{ELP}_{E_0}(0),$$
    \boldsymbol{ELP}_{E_2}(0),$$
    \boldsymbol{ELP}_{E_1}(0)\}$, which is infeasible; The second GLP to try is $\boldsymbol{GLP}_1$= $\{\boldsymbol{ELP}_{E_0}(0),$ $\boldsymbol{ELP}_{E_2}(0),$$\boldsymbol{ELP}_{E_1}(1)\}$, which is feasible. Therefore, any GLPs smaller than $\boldsymbol{GLP}_1$ will be skipped based on Theorem~\ref{theorem_bound_single_chain}.
    After that, only 3 GLPs have to be evaluated: 
    $\boldsymbol{GLP}_2=$ 
    $\{$\elpmore{0}{1}, \elpmore{2}{0}, \elpmore{1}{0}$\}$, 
    $\boldsymbol{GLP}_3=\{$\elpmore{0}{0}, \elpmore{2}{1}, \elpmore{1}{0}$\}$
    and $\boldsymbol{GLP}_4$$=\{$\elpmore{0}{1}, \elpmore{2}{1}, \elpmore{1}{0}$\}$.
    However, neither is feasible. 
    As a result, the number of complete GLPs to evaluate decreases from 36 in Example~\ref{example_gcp_bf} to only 1, although there are 4 infeasible GLPs to skip.
\end{Example}

\begin{theorem}
 \label{theorem_skip_perf_bound}
 The difference between solutions found by Algorithm~\ref{alg_opt_skipped} and the optimal solution to problem~\ref{obj_overall} is upper-bounded by $\sum_{\mathcal{C} \in \boldsymbol{\mathcal{C}}} B_{\mathcal{C}}$, where $B_{\mathcal{C}}$ is given by Lemma~\ref{theorem_compare_mn_bound}.
\end{theorem}
\begin{proof}
We prove the theorem under data age optimization. Reaction time optimization can be performed similarly.
    Let's denote the optimal GLP found by Algorithm~\ref{alg_opt_skipped} as $\boldsymbol{GLP}^{Alg\ref{alg_opt_skipped}*}$, the optimal GLP to problem~\ref{obj_overall} as $\boldsymbol{GLP}^{*}$. Furthermore, we divide all the possible complete GLPs into two sets: $\boldsymbol{S}^{eval}$ and $\boldsymbol{S}^{skip}$, which denote the complete GLPs evaluated during Algorithm~\ref{alg_opt_skipped} and those skipped. We know $\boldsymbol{GLP}^{Alg\ref{alg_opt_skipped}*} \in \boldsymbol{S}^{eval}$. If $\boldsymbol{GLP}^{*} \in \boldsymbol{S}^{eval}$, then Algorithm~\ref{alg_opt_skipped} finds the optimal solutions because Algorithm~\ref{alg_opt_skipped} will select the best GLPs within $\boldsymbol{S}^{eval}$.

    If $\boldsymbol{GLP}^{*} \in \boldsymbol{S}^{skip}$, then $\boldsymbol{GLP}^{*}$ is skipped either based on Theorem~\ref{theorem_bound_single_chain} (performance bound is provided by Theorem~\ref{theorem_bound_multi_chain}), or based on \alglineref{worsePatterns} in Algorithm~\ref{alg_opt_skipped} ($\boldsymbol{GLP}^{*}$ cannot achieve better performance than $\boldsymbol{GLP}^{Alg\ref{alg_opt_skipped}*}$). Therefore, the theorem is proved. 
\end{proof}

\subsection{Computation Efficiency}
A simple worst-case complexity analysis when optimizing the data age or time disparity is given as follows:
\begin{equation}
    O(\textbf{E})= {\textstyle \prod_{E \in \boldsymbol{E}}} \ |\textbf{ELP}_{E}|
\end{equation}
where $ \boldsymbol{E}$ denotes the edges that appear in the objective function, $|\cdot|$ denotes the size of a set. 
Reaction time optimization has a similar complexity.
The analysis is pessimistic because it is difficult to analyze the improvements brought by the symbolic operation or the back-tracking algorithms, though they can usually bring significant speed-up. 

\section{Generalizations and Limitations}
\label{section_application}
\subsection{Sporadic Tasks Optimization}

Our optimization algorithms can work with cause-effect chains that contain sporadic tasks, \ie tasks released non-periodically.
With the help of the compositional theorem proposed in~\cite{Gnzel2021TimingAO}, the cause-effect chain can be separated into several sub-chains where some sub-chains do not contain periodic tasks. The sub-chains with only periodic tasks can still be optimized with the algorithm proposed in this paper. However, notice that evaluating and optimizing time disparity for sporadic tasks may not be easily applicable.

\subsection{Other Objective Functions and Schedulability Analysis}
\label{section_more_application_schedulability}

The backtracking algorithm in Section~\ref{section_opt_skip} works with many types of objective functions and their combination. 
If these objective functions can be transformed into linear functions such as DART, then optimality is guaranteed with good runtime speed; 
otherwise, the algorithm may still be applicable (e.g., time disparity jitter), though without the optimality guarantee. 
In the latter case, nonlinear optimization algorithms~\cite{Wang2023RTAS, Wang2024AGA} may also be considered.

\subsection{Pessimistic RTA and Second Step Optimization}
\label{sec: limitation_solution}
It is difficult to directly adopt the exact schedulability analysis~\cite{Tindell1994ADDINGTT, Palencia1998SchedulabilityAF, Redell2002CalculatingEW} when optimizing the virtual offset and virtual deadline because of their discrete and nonlinear forms. Therefore, we can only guarantee optimality with the schedulability analysis method used during optimization.

A potential solution is performing a second optimization step to optimize only the virtual offset based on the exact response time. 
The exact response time only depends on the virtual offset and can usually be easily obtained. Therefore, we can utilize the exact response time to optimize the virtual deadline to improve performance further.
A limitation is that it can only find optimal solutions within the current GCP constraints.
Therefore, in experiments, we utilize the heuristic proposed in Maia~\etal~\cite{Maia2023ReducingEL} for the second optimization step if the objective functions are reaction time or data age; 
we can optimize the virtual deadline following \defref{glp evaluate} for time disparity and jitter optimization.


\section{Experimental Results and Discussion}
\label{experiment_section}
The optimization framework was implemented in C++ and tested on a computing cluster (AMD EPYC 7702 CPU). 
The following baseline methods are roughly ordered from least to most effective based on some experiment performance:
\begin{itemize}[leftmargin=*]
    \item DefLET, the default LET model.
    \item Martinez18, an offset optimization method~\cite{Martinez2018AnalyticalCO} for LET.
    \item Bradatsch16, it shrinks the length of LET interval to the worst-case response time; the offset is $0$~\cite{Bradatsch2016DataAD}.
    \item Implicit, implicit communication following~\cite{Hamann2017CommunicationCD}.
    \item Maia23, it uses the smallest relative start time and biggest relative finish time of each task as the LET interval~\cite{Maia2023ReducingEL}. 
    We did not compare their JLD optimization algorithm because it changes the scheduling algorithm.
    \item fLET$\_$GCP\_Enum, from Section \ref{section_opt_bf}.
    \item fLET$\_$GCP\_Backtracking\_LP, from Section~\ref{section_opt_skip}. 
    \item fLET$\_$GCP\_SymbOpt, Algorithm~\ref{alg_opt_skipped}.
    \item fLET$\_$GCP\_Extra, from \secref{limitation_solution}.
\end{itemize}
All the involved LP problems are solved by CPLEX~\cite{cplex2009v12}. The full experiments can be reproduced following the repository: \href{https://github.com/zephyr06/LET_OPT}{$\text{https:github.com/zephyr06/LET\_OPT}$}.

\emph{DBP is not included in our comparison, as it is proven to perform worse than implicit communication theoretically~\cite{Tang2023ComparingCP}.}

In experiments, the task set is scheduled by the rate-monotonic algorithm. 
A safe response time analysis (RTA) is used to obtain the response time $R_i$ for task $\tau_i$~\cite{Joseph1986FindingRT}:
\begin{equation}
    R_i = C_i +\sum_{j \in \text{hp}(i)} \ceil{\frac{R_i}{T_j}}{C_j}
    \label{rta_LL}
\end{equation}
where $\text{hp}(i)$ denotes the tasks with higher priority than $\tau_i$. 

The time limit for optimizing one task set is 1000 seconds. If not time out, Martinez18 finds optimal offset assignments; fLET\_GCP optimization algorithms find the optimal virtual offset and virtual deadline with respect to the RTA~\eqref{rta_LL}.

To maximize the chance of finding a feasible solution within a limited time, the GCP optimization algorithms first evaluate the GCP of the default LET; after that, they follow the searching order mentioned in Section~\ref{section_dart_symbopt}.

\begin{figure*}
     \centering
        \includegraphics[width=0.95\linewidth]{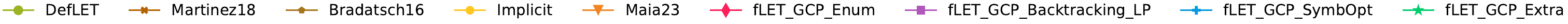}
     \vspace{-0.2cm} 
\end{figure*}
\begin{figure*}[ht!]
\centering 
 \begin{subfigure}[t]{0.32\textwidth} 
    \centering
    \includegraphics[width=\textwidth]{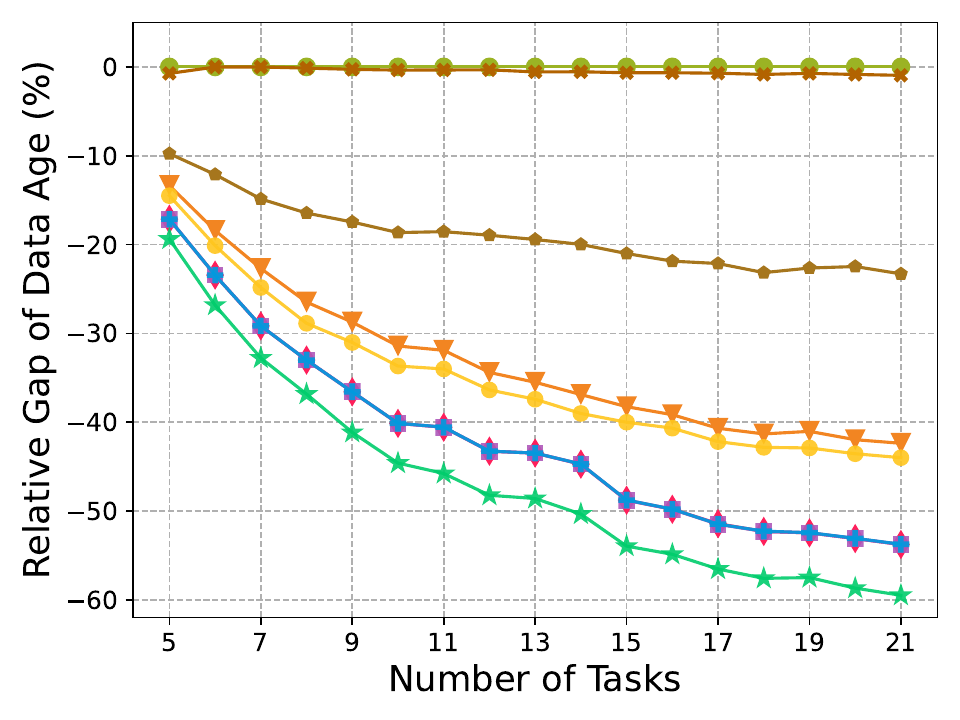}
        \caption{Data Age, Single Chain, High Utilization}
    \label{fig_da_perf}
\end{subfigure}
\begin{subfigure}[t]{0.32\textwidth}
    \centering
    \includegraphics[width=\textwidth]{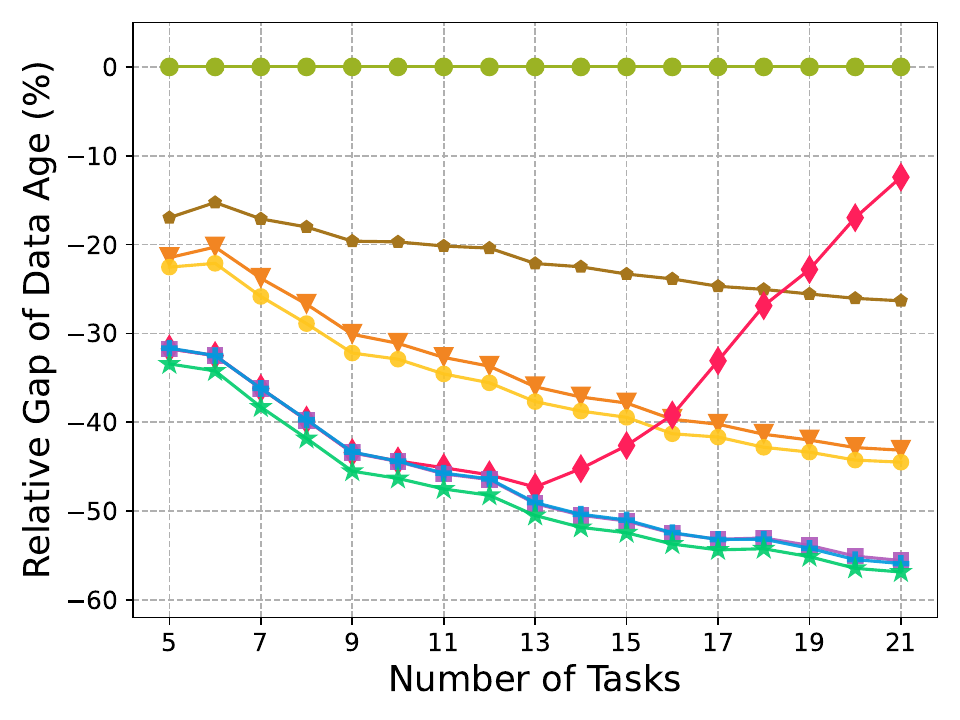}
    \caption{Data Age, Multiple Chains}
    \label{fig_da_perf_3chains}
\end{subfigure}
\begin{subfigure}[t]{0.32\textwidth} 
    \centering
    \includegraphics[width=\textwidth]{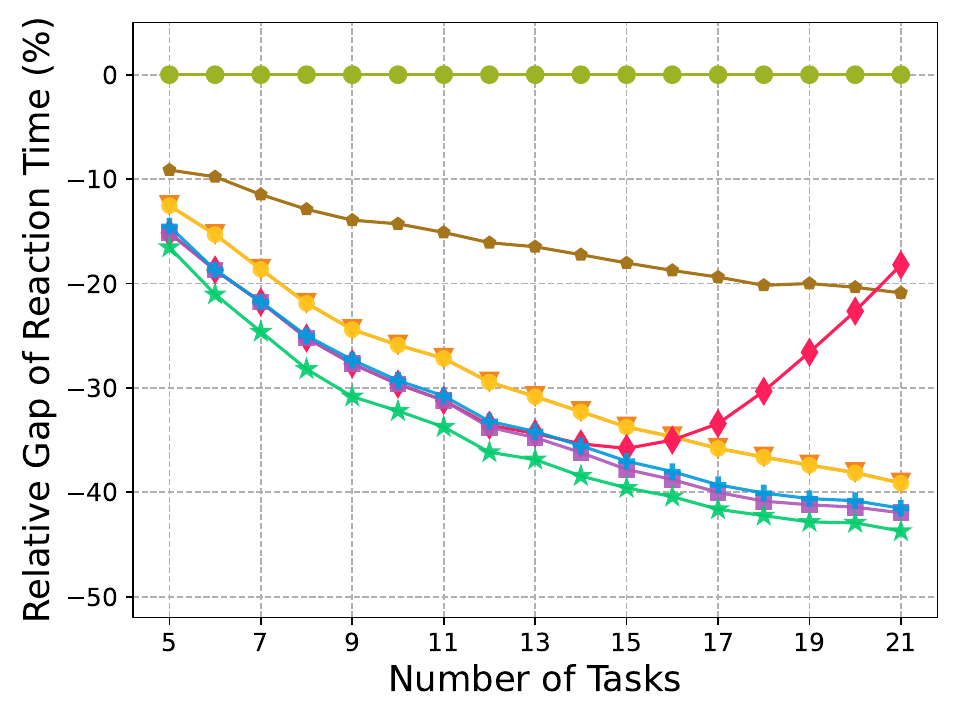}
        \caption{Reaction Time, High Utilization}
    \label{fig_rt_perf}
\end{subfigure}

\begin{subfigure}[c]{0.32\textwidth}
    \centering
    \includegraphics[width=\textwidth]{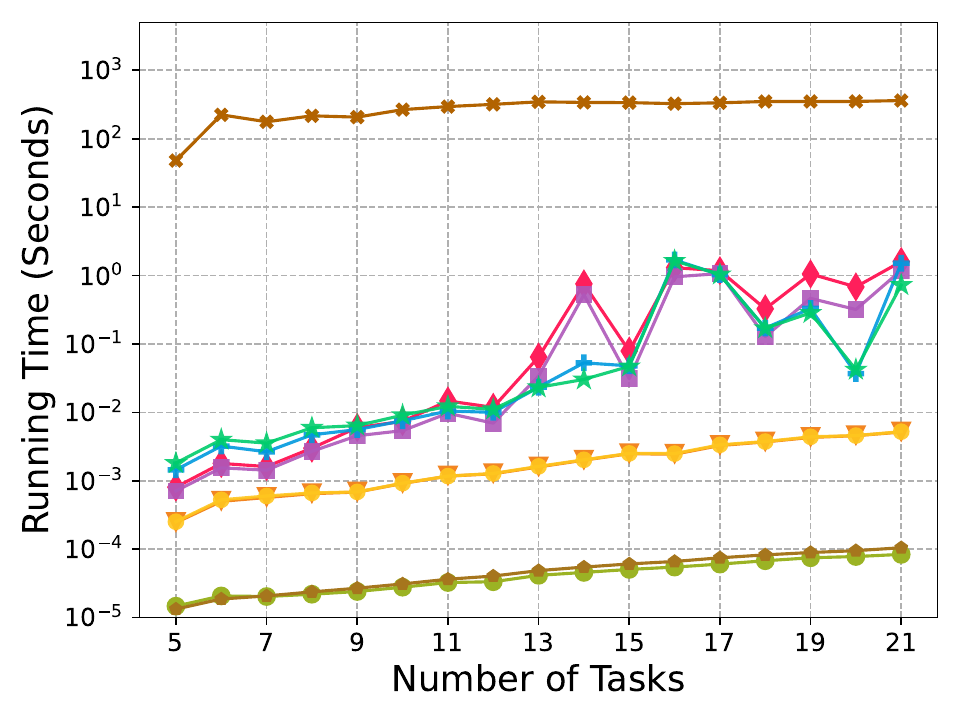}
          \caption{Data Age, Single Chain, High Utilization}
    \label{fig_da_time}
\end{subfigure}
\begin{subfigure}[c]{0.32\textwidth}
    \centering
    \includegraphics[width=\textwidth]{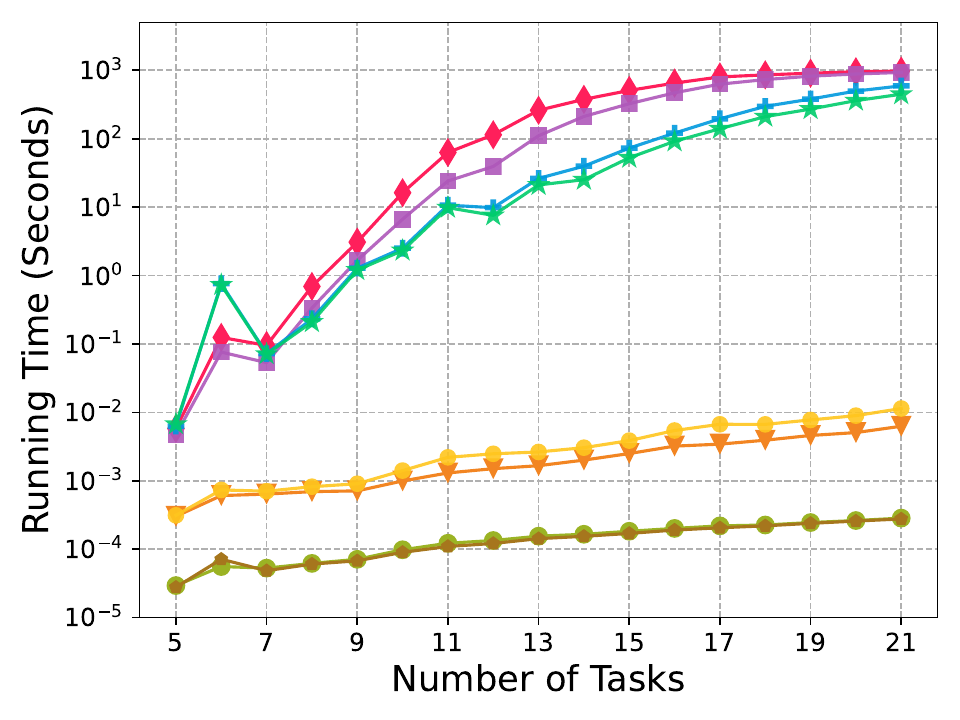}
    \caption{Data Age, Multiple Chains}
    \label{fig_da_time_3chains}
\end{subfigure}
\begin{subfigure}[c]{0.32\textwidth}
    \centering
    \includegraphics[width=\textwidth]{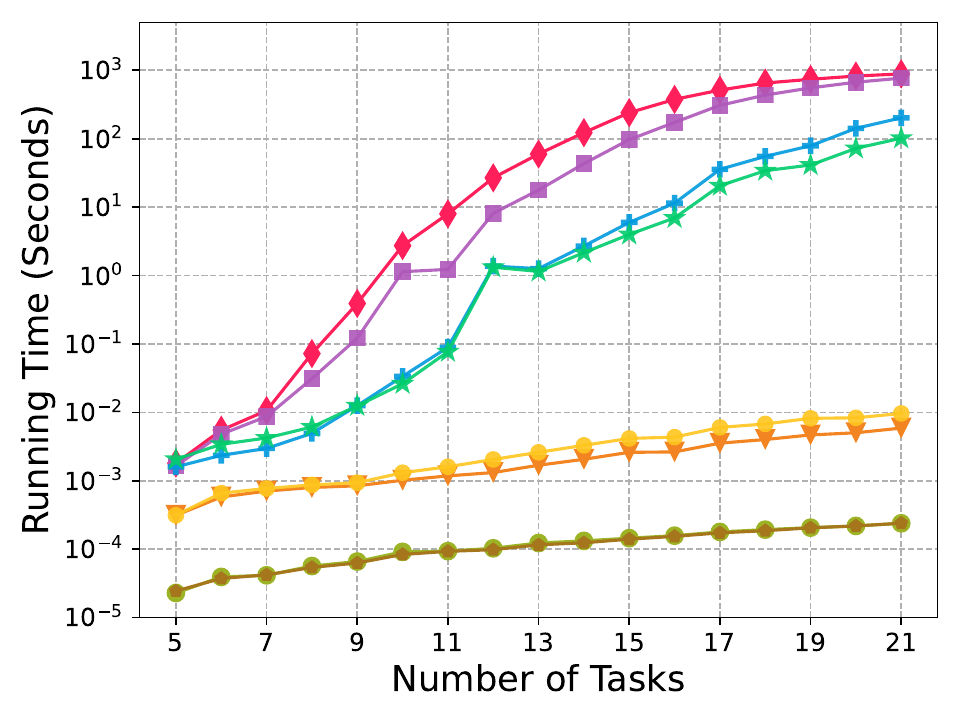}
        \caption{Reaction Time, High Utilization}
    \label{fig_rt_time}
\end{subfigure}

\begin{subfigure}[c]{0.32\textwidth}
    \centering
    \includegraphics[width=\textwidth]{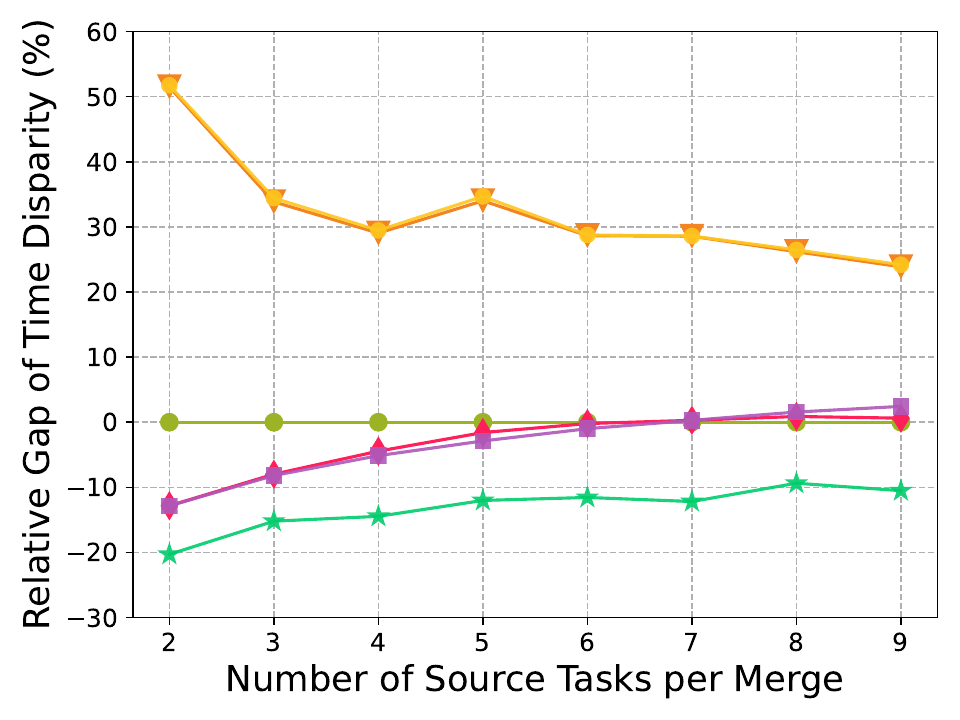}
          \caption{Time Disparity and Jitter}
    \label{fig_sf_obj}
\end{subfigure}
\begin{subfigure}[c]{0.32\textwidth}
    \centering
    \includegraphics[width=\textwidth]{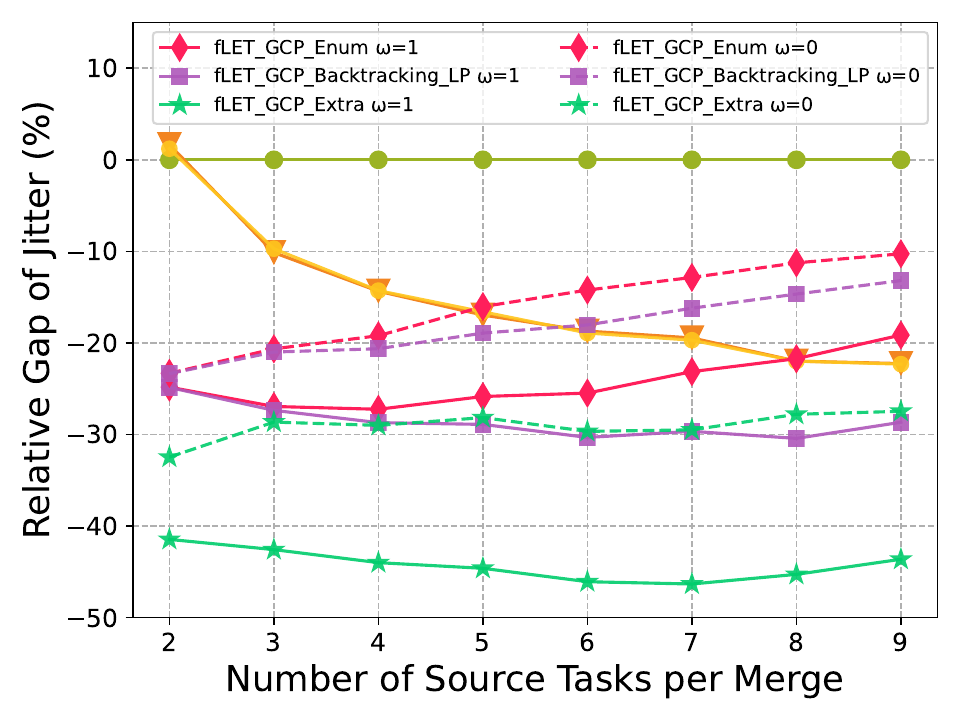}
    \caption{Time Disparity and Jitter}
    \label{fig_sf_jitter}
\end{subfigure}
\begin{subfigure}[c]{0.32\textwidth}
    \centering
    \includegraphics[width=\textwidth]{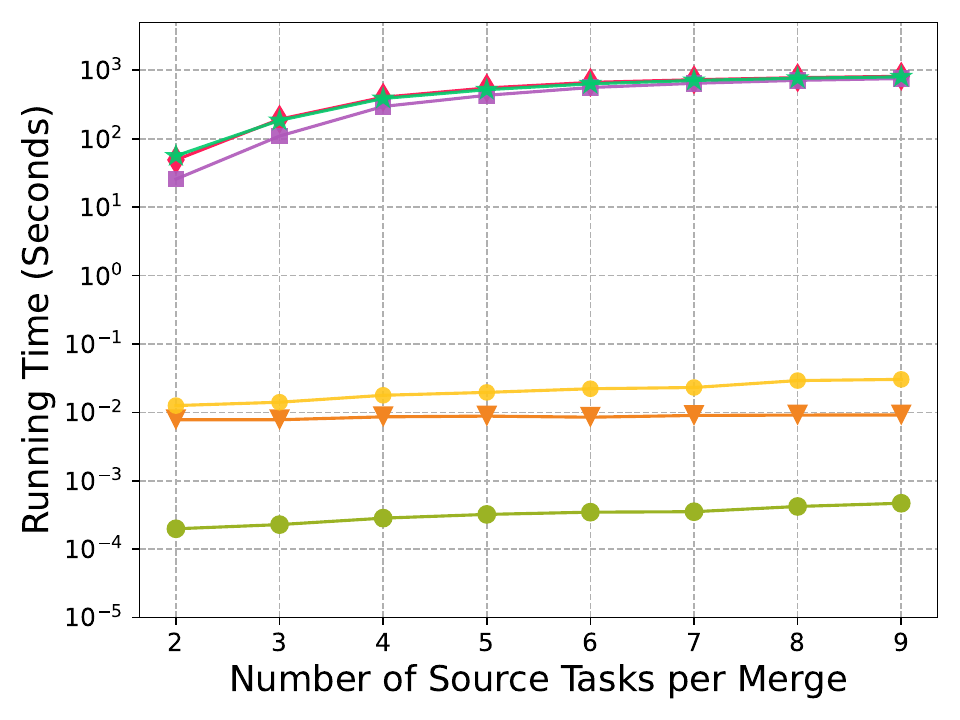}
    \caption{Time Disparity and Jitter}
    \label{fig_sf_time}
\end{subfigure}
\caption{Relative performance gap and runtime (log scale) when optimizing different performance metrics individually. fLET outperforms other communication protocols and state-of-the-art optimization algorithms while offering broader applicability. 
 } 
\label{figs_exp_results_all}
\end{figure*}

\subsection{Task Set Generation}
The DAG task sets are generated following the WATERS Industry Challenge~\cite{Kramer15benchmark}. Task periods are randomly generated from a predefined set $\{1, 2, 5, 10, 20, 50, 100, 200, 1000\}$ whose relative probability distribution $\{$3, 2, 2, 25, 25, 3, 20, 1, 4$\}$~\cite{Kramer15benchmark}. 
The task set utilization is randomly selected from $[0.5, 0.9] \times m$, where $m$ is the number of cores ($m=4$ in our case). 
In high-utilization task sets, the utilization is always $0.9m$.
Each task $\tau_i$'s execution time $C_i$ is generated by Uunifast~\cite{Bini2005MeasuringTP} except that each task's utilization is no larger than 100\%. The deadline $D_i^{org}$ is set as the period $T_i$.  

The DAG structure is generated following He \etal~\cite{He2021ResponseTB}. 
Random edges are added from one task to another with 0.9 probabilities (task sets with many cause-effect chains cannot be generated with smaller probabilities).  
After generating the DAG, we randomly select source and sink nodes and use the shortest path algorithm by Boost Graph Library~\cite{Siek2001TheBG} to generate random cause-effect chains. The length and the activation pattern distributions of the cause-effect chains follow Table~\RomanNumeralCaps{6} and Table~\RomanNumeralCaps{7} in Kramer~\etal~\cite{Kramer15benchmark}. 
To generate many random task sets while simultaneously satisfying the two distribution requirements, we first generate many random task sets with random cause-effect chains, calculate the likelihood for each task set, and then sample 1000 random task sets weighted by the likelihood for each $N$. The task sets used in our experiment follow a similar distribution pattern as in Kramer~\etal~\cite{Kramer15benchmark}. 
 
We generate 1000 random task sets for a given number of tasks within a task set. 
Each task set $\boldsymbol{\tau}$ of $N$ tasks has $1.5N$ to $3N$ random cause-effect chains. 
For example, there are 31 to 63 cause-effect chains when there are 21 tasks in the task set (following Hamann~\etal~\cite{waters2017}).
All the generated task sets are schedulable based on RM.

There are always 21 tasks in each task set when optimizing the time disparity and jitter. However, the maximum number of source tasks varies from $2$ to $9$ following ROS~\cite{Li2022WorstCaseTD}. 
Besides, we randomly select $1$ to $4$ merges to constitute the objective function~\eqref{obj_overall_sf}. The weight in the objective function~\eqref{obj_overall_sf} is $1$.

\begin{table*}[!ht]
\centering
\caption{Autonomous robot case study results}
\label{tableCaseStudyResult}
\begin{tabular}{@{}lllllllll@{}}
\toprule
                   & DefLET & Martinez18 & Bradatsch16 & Implicit & Maia23 & fLET\_Enum    & fLET$\_$Backtracking & fLET\_SymbOpt \\ \midrule
Reaction Time (ms) & 4040   & N/A        & 3237        & 3237     & 3237   & \textbf{2725} & \textbf{2725}        & \textbf{2725} \\
Data Age (ms)      & 5000   & 5000       & 4197        & 4197     & 4197   & \textbf{3685} & \textbf{3685}        & \textbf{3685} \\
Sensor Fusion, Jitter (ms) & 1500, 1500 & N/A & N/A & 1712, 1500 & N/A & \textbf{1461, 1422} & \textbf{1461, 1422} & \textbf{1461, 1422} \\ \bottomrule
\end{tabular}%
\end{table*}

Since finding optimal time disparity jitter is computationally expensive (Observation~\ref{lemma_td_jitter_linearizable}), we optimize only time disparity when evaluating a GLP (\defref{glp evaluate}) for better efficiency in experiments. We then select the GLP with minimum time disparity and weighted jitter as the final solution. Therefore, optimality for the objective function~\eqref{obj_overall_sf} is not guaranteed, though it is possible by solving mixed-integer programming.

Fig.~\ref{figs_exp_results_all} shows the runtime speed and the performance gap of a method against the default LET, which is defined as 
\begin{equation}
    \frac{\mathcal{F}_{method} - \mathcal{F}_{defLET}}{\mathcal{F}_{defLET}}\times 100 \%
\end{equation}
Reaction time and data age optimization results are very similar, so we only show one of them under the same situation.

\subsection{Autonomous Robot Case Study}

We performed a realistic case study following the autonomous robot system built in Sifat \etal~\cite{Sifat2023ASM}. The computation tasks and their dependency relationship are shown in Fig.~\ref{fig_dag_case_study}. The tasks are executed in a real embedded system (NVIDIA Jetson AGX Xavier), and their WCET and period are shown in Table.~\ref{tableRobotComputingTask}. The tasks adopt implicit deadlines. The end-to-end latency is measured for the critical cause-effect chain: SLAM $\rightarrow$ Path Planning $\rightarrow$ Control. The original paper~\cite{Sifat2023ASM} does not consider time disparity, so we selected a reasonable merge (Depth Estimation and Path Planning to Control) to conduct experiments on time disparity and jitter. The results are shown in Table.~\ref{tableCaseStudyResult}.

\begin{figure}[!ht]
\centering
\includegraphics[width=1.0\columnwidth]{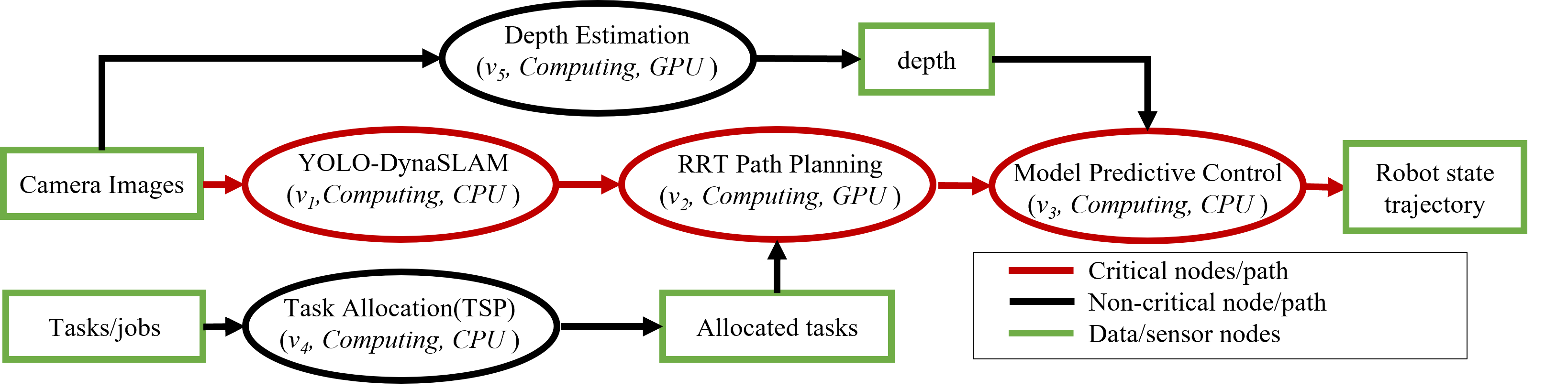}
\caption{
Autonomous robot tasks' dependency graph, \cite{Sifat2023ASM}.
}
\label{fig_dag_case_study}
\end{figure}

\begin{table}[ht!]
\centering
\caption{Autonomous robot computing tasks}
\label{tableRobotComputingTask}
\begin{tabular}{@{}lll@{}}
\toprule
Task             & Period (ms) & WCET (ms) \\ \midrule
SLAM             & 1000        & 500       \\
Path Planning    & 2000        & 1188      \\
Control          & 40          & 37        \\
Task Allocation  & 10000       & 10000     \\
Depth Estimation & 500         & 400       \\ \bottomrule
\end{tabular}%
\end{table}

\subsection{Performance Analysis}

\subsubsection{fLET vs Other Communication Mechanisms}
The performance improvements and the extra time determinism benefits make fLET more competitive than other communication mechanisms, such as default LET and implicit communication, especially when the end-to-end latency metrics are important. Furthermore, fLET maintains the extra time determinism benefits that implicit communication does not have.

\subsubsection{fLET Optimization Algorithms} 
In experiments, the suboptimality gap in Theorem~\ref{theorem_skip_perf_bound} is typically under 1\%. Meanwhile, the symbolic operation brings 2x to 30x speed-up compared with the backtracking algorithm. The speed advantages are crucial for large task sets or situations with tight time budgets because the algorithm performance may otherwise seriously degrade (\eg  Martinez18 and fLET$\_$GCP\_Enum). 

\subsubsection{Multi-objective optimization}
In Fig.~\ref{fig_sf_obj}, the time disparity results nearly overlap between $w=0$ and $w=1$, so we only show the results for $w=1$.
Fig.~\ref{fig_sf_obj} and~\ref{fig_sf_jitter} show that the fLET optimization algorithms effectively balance conflicting goals (time disparity and jitter). 
Although they cannot guarantee optimality when jitter is present in the objective functions, fLET still significantly outperforms the baseline methods.
It is also interesting to note that optimizing the data age often inherently optimizes the reaction time and vice versa. This observation is consistent with the recent work~\cite{Gnzel2023OnTE}, which proves that the maximum reaction time is equivalent to the maximum data age (note that they use a slightly different definition than those introduced in Section~\ref{sec: system_model}).

\subsubsection{Pessimistic vs exact RTA} 
Performing extra optimization with the exact RTA brings big performance improvements when optimizing time disparity and jitter (TDJ). In contrast, the improvements when optimizing DART are more limited. This is probably because TDJ is more ``nonlinear'' than DART (Theorem~\ref{lemma_td_jitter_linearizable}) and more challenging to optimize.

\subsubsection{Time-out issue}
Despite the possibility for timeouts as systems scale, Fig.~\ref{figs_exp_results_all} shows that substantial performance improvements are still attainable even if the fLET optimization algorithms did not finish within the time limits.

\section{Conclusions and Future Work}
\label{conclusion_section}
In this paper, we propose novel optimization algorithms to optimize both the reading and writing time of the LET model, improving the system performance with dataflow determinism preserved. 
In particular, the optimization problem is reformulated to optimize communication patterns, which can be solved efficiently with LP while supporting the optimization of many metrics, such as end-to-end latency, time disparity, and its jitter.
We also introduce novel symbolic operations and prove bounded suboptimality when minimizing the worst-case data age or reaction time.
Experimental results highlight significant performance improvement over other communication mechanisms (implicit communication and DBP) and state-of-the-art LET extensions.

\newpage
\section{ACKNOWLEDGMENT}
This work is partially supported by NSF Grants No. 1812963 and 1932074. We also sincerely appreciate the anonymous reviewers and the shepherd whose insightful feedback and guidance significantly contributed to the refinement and enhancement of this work.

\bibliographystyle{ieeetr}
\bibliography{scheduling} 

\end{document}